\documentclass[11pt]{article}
\RequirePackage[OT1]{fontenc}
\RequirePackage[colorlinks,citecolor=blue,urlcolor=blue]{hyperref}
\RequirePackage{amsthm,amsmath,booktabs}
\RequirePackage{natbib}
\usepackage{diagbox}
\usepackage{amssymb,bbm,tikz}
\usepackage{mathtools}
\mathtoolsset{showonlyrefs}
\usepackage{cases}
\usepackage{graphicx}
\usepackage{tabularx}
\usepackage{microtype}
\usepackage{enumitem}
\usepackage{color, colortbl}
\definecolor{Gray}{gray}{0.9}
\usepackage{url}
\usepackage{amsfonts, mathabx}
\usepackage[ruled,linesnumbered, vlined, noend]{algorithm2e}

\RequirePackage[colorlinks,citecolor=blue,urlcolor=blue]{hyperref}
\newcounter{myalgctr}
\newenvironment{rem}{%      define a custom environment
   \vskip1mm\indent%         create a vertical offset to previous material
   \refstepcounter{myalgctr}% increment the environment's counter
   \textbf{Remark \themyalgctr}% or \textbf, \textit, ...
   }{\hfill$\diamond$\par}  %          create a vertical offset to following material
\numberwithin{myalgctr}{section}

\usepackage{appendix,authblk}

\let\hat\widehat

\newtheorem{thm}{Theorem}

\newtheorem*{definition*}{Definition}
\newtheorem*{remark*}{Remark}
\makeatletter
\def\namedlabel#1#2{\begingroup
    #2%
    \def\@currentlabel{#2}%
    \phantomsection\label{#1}\endgroup
}
\makeatother

 \title{\bf Nested Conformal Prediction Sets for Classification with Applications to Probation Data}
\author[1]{Arun K. Kuchibhotla\thanks{Email: {\tt arunku@cmu.edu}.}}
 \author[2]{Richard A. Berk\thanks{Email: {\tt berkr@sas.upenn.edu}.}} 
\affil[1]{Carnegie Mellon University }
 \affil[2]{University of Pennsylvania}

\begin{document}
\maketitle

\begin{abstract}
Risk assessments to help inform criminal justice decisions have been used in the United States since the 1920s. Over the past several years, statistical learning risk algorithms have been introduced amid much controversy about fairness, transparency and accuracy. In this paper, we focus on accuracy for a large department of probation and parole that is considering a major revision of its current, statistical learning risk methods. Because the content of each offender's supervision is substantially shaped by a forecast of subsequent conduct, forecasts have real consequences. Here we consider the probability that risk forecasts are correct. We augment standard statistical learning estimates of forecasting uncertainty (i.e., confusion tables) with uncertainty estimates from nested conformal prediction sets.  In a demonstration of concept using data from the department of probation and parole, we show that the standard uncertainty measures and uncertainty measures from nested conformal prediction sets can differ dramatically in concept and output. We also provide a modification of nested conformal called the localized conformal method to match confusion tables more closely when possible.  A strong case can be made favoring the nested and localized conformal approach. As best we can tell, our formulation of such comparisons and consequent recommendations is novel.
\end{abstract}

% \begin{keyword}[class=MSC]
% \kwd{62P25}
% \kwd{62G07}
% \kwd{62H30}
% \kwd{62Q99}
% % \kwd{62F25}
% % \kwd{60F05}
% %\kwd[; secondary ]{}
% \end{keyword}
\noindent%
 {\bf Keywords} 
Conformal Prediction, Level Sets, Classification, Risk Assessment, Criminal Justice, Statistical Learning

% \end{frontmatter}

\section{Introduction}
Risk assessments have been used in the United States to help inform criminal justice decisions since the 1920s~\citep{burgess1928factors}. Typically, these risk assessments forecast whether a given individual, already convicted of a crime, will re-offend. Over the last decade, statistical learning assessments of risk have been introduced in some jurisdictions, often with considerable controversy. Concerns include fairness, transparency, and accuracy. All three are important, but here we focus on accuracy. Excellent treatments of fairness and transparency can be found elsewhere~\citep{carlson2017need,berk2019machine,coglianese2019transparency,huq2018racial,kearns2019ethical,rudin2age}. 
 
In this paper, we analyze proprietary data from a department of probation and parole in a large metropolitan area that has for several years used a statistical learning classifier to assess risk.\footnote
{
We will focus exclusively on probationers who constitute the vast majority of offenders supervised. Parolees are generally supervised at the state level. 
}
As described below, the risk procedure has performed well, but a major revision with current data is under consideration. 

Forecasting accuracy used to justify statistical learning for risk assessment conventionally has been estimated out-of-sample from confusion tables. For example, among offenders forecasted to re-offend, an estimated probability is meant to convey the chances that the projected recidivism will occur. This is an aggregate assessment that at best is an indirect indicator of how a forecast will perform for a particular offender. Moreover, conditioning on the forecast
imposes constraints that neglect important sources of uncertainty. 

Conformal prediction~\citep{vovk2005algorithmic,lei2018distribution,gupta2019nested,romano2020classification,angelopoulos2020uncertainty} is an alternative, generic methodology providing accuracy guarantees for predictions from any machine learning algorithm. In this paper, we emphasize categorical response variables forecasted via nested conformal prediction sets \citep{gupta2019nested}. In conformal prediction, each possible set of outcome classes is evaluated to find the ``best'' prediction set. Unlike forecasting procedures that  provide the best single outcome, conformal prediction can return a set of one or more outcome classes. Associated with the best prediction set is a estimate of the probability that the chosen prediction set contains the true outcome. Such claims have finite sample guarantees. One also can construct the best prediction set for particular offender types, defined by their predictor values. Uncertainty claims for these prediction sets have asymptotic guarantees.

Forecasts and associated uncertainties produced by common statistical learning procedures can differ substantially from the nested conformal approach. Such differences are highlighted in the pages ahead. We see this paper as a demonstration of concept as plans for a new statistical learning, risk assessment, procedure are formulated.

In Section~\ref{sec:policy-description}, we describe the policy setting and introduce the data from which forecasts will be made that illustrate the principal statistical issues. In Section~\ref{sec:nested-conformal-prediction}, we provide some background on conformal prediction and nested conformal prediction for numeric $Y$. In Section~\ref{sec:categorical-Y}, we discuss nested conformal prediction for categorical $Y$, which is of primary interest for our application. We also compare nested conformal prediction with confusion tables and address the discrepancies. In Section~\ref{sec:local-conformal-prediction}, we offer an alternative nested conformal method for categorical $Y$ that can corroborate inferences from a confusion table. Finally, Section~\ref{sec:conclusions} provides several overall conclusions and makes recommendations for practice.

\paragraph{Notation} Throughout the paper, we use $1 - \alpha$ to denote the required probability guarantee as in~\eqref{eq:unconditional-guarantee} and~\eqref{eq:conditional-guarantee}. We will also use $\gamma\in[0, 1]$ to denote intermediate probabilities. For any event $A$, we write $\mathbbm{1}\{A\}$ to denote the indicator, i.e., $\mathbbm{1}\{A\} = 1$ if event $A$ occurs/holds and $\mathbbm{1}\{A\} = 0$ if event $A$ does not occur.
% \begin{center}
%     {\color{red}[To be populated with more notations in the paper.]}
% \end{center}

\section{The Policy Setting}
\label{sec:policy-description}

A sentence of probation ordered by a court typically is an alternative to incarceration. A convicted offender is released under supervision, the intensity and content of which can vary widely. For offenders thought to pose a serious risk to public safety or likely to re-offend, the supervision can be very intrusive: home detention, close surveillance by probation officer, unannounced home or workplace visits, frequent drug tests, GPS monitoring, and more. For offenders seen as low risk, supervision can be in name only. In the extreme, a probationer may only need to report once a month to an electronic kiosk to update a home address, contact phone number, and place of work. However, attached to any sentence of probation can be a variety of conditions such as attending AA meetings, anger management counseling, curfews, community service, restitution, and prohibitions from socializing with certain individuals such as gang members. 

Agencies in charge of probation supervision often have wide discretion in how oversight is implemented. Some probation conditions can be relaxed or even removed. New probation conditions can be imposed. Probation officers also have considerable discretion in how strictly the conditions are enforced. For example, a probation officer may choose to overlook one failure to attend an office meeting but not several over 6 months. 

Although there can be good reasons for discretion exercised by probation departments and individual probation officers, there are the dual risks of inefficiencies and inconsistencies across probationers. In response to these and other concerns, a large, urban, department of probation and parole began in 2007 using random forests to inform decisions about how to allocate the intensity and content of probation supervision. Their intent was to provide the least restrictive forms of supervision consistent with the public safety and control of crime. A necessary condition was determining the varying needs and risks for each probationer \citep{ahlman2008appd}. 

The algorithm's forecasting accuracy was promising in test data, and approximately 60\% of the offenders were projected to be arrest-free while under supervision. But by itself, that was insufficient \citep{barnes2010low}. The implications for agency practices were unclear.  

To clarify matters, a randomized field experiment was undertaken to determine if offenders forecasted to be arrest-free could supervised at ``Low-Intensity'' uncompromised by an increase in re-arrests.  1,559 low risk offenders just beginning their probation sentence were identified by the algorithm. Half were randomly assigned to standard supervision and half to were randomly assigned to low intensity supervision. 

Under the low intensity supervision, case loads for each probation officer were increased on average from approximately 150 offenders to 400 offenders, which automatically meant less frequent and lengthy contacts between probationers and probation officers. For example, the number of office visits per year was reduced on average from 4.5 to 2.4. Drug testing was only to be done if court ordered. After three failed drug tests, drug treatment was offered as an alternative to incarceration. 

All of the offenders in the study were followed for one year through their administrative records. Any arrests were recorded. A standard analysis of the experiment showed that the re-arrest proportions were effectively the same for offenders under either supervision regime \citep{barnes2010low}. A second analysis using different methods arrived at the same conclusion \citep{berk2010second}. 

As a result, the department reorganized its practices so that probationers forecasted to be arrest-free were placed in low intensity supervision. Cost savings were allocated to more intensive and expensive supervision for probationers forecasted to be re-arrested. This made the machine learning risk assessments critical. A forecasted outcome class could have dramatic effects on an offender's experience while on probation and impact the agency's budget. 

Over the past decade, the statistical learning risk instrument was retrained twice with more recent data, but there were no material changes in the overall approach. For example, accuracy was still addressed using confusion tables with test data. Over the past year, there has been growing interest in major revisions of the risk instrument. Fairness has been addressed at some length in other work \citep{berk2020using,berk2020improving}. We focus on forecasting accuracy, while making a distinction between the probability that a given forecasted class is correct and the probability that a ``best'' prediction set is correct.
\medskip

\noindent\textbf{Description of the Data.}
% \section{Probation Risk Data To be Analyzed}
% \label{sec:data}
Building on our discussion in \ref{sec:policy-description}, we have data on 102,555 offenders who began a sentence of probation between 2009 and 2013 in a large urban jurisdiction. The substantial number of observations means that asymptotic as well as finite sample guarantees are available. For each offender, there are 2 years of followup data characterizing conduct while on probation. At the urging of probation officials and other stakeholders, we use three outcome variable classes: an arrest for a crime of violence, an arrest for a nonviolent crime, and no arrest at all. Concerns about crime typically emphasize violent crime. Distinctions between crimes of violence and nonviolent crimes are critical because public safety is an essential governmental responsibility.  

Seventeen predictors were taken from the usual criminal justice administrative records: rap sheets, charges for the crime(s) that led to the probation sentence, age, gender and the earliest crime for which the offender was charged as an adult. These are the predictors.

\begin{itemize}
\item
Age: in years
\item
Gender: male=1, female=0
\item
Elapsed time since most recent arrest for a violent crime: in years
\item
Elapsed time since most recent arrest for a weapons crime: in years
\item
Age at earliest arrest: in years
\item
Age at earliest arrest for a crime of violence: in years
\item
Prior jail Incarcerations: a count
\item
Prior prison incarcerations: a count
\item
Prior murder arrests: a count
\item
Prior property crime arrests: a count
\item
Prior violent crime arrests: a count
\item
Prior drug distribution arrests: a count
\item
Prior weapons arrests: a count
\item
Current charges for a violence crime: a count
\item
Current charges for a drug crime: a count
\item
Current charges for a property crime: a count
\item
Current charges for a firearm crime: a count
\end{itemize}

The data are seen as realized from a joint probability distribution characterizing all offenders in the relevant jurisdiction sentenced to probation from several years before 2009 to several years after 2013. During that time, there were no major changes in law or administrative procedures that would have materially altered the mix of offenders sentenced or sentencing considerations by local judges. A case for exchangeability is made by noting that data on each offender was effectively realized independently. Whether one individual was placed on probation was arguably unrelated to whether any other individual was placed on probation. For example, cases are tried one at a time and in this jurisdiction, there are many trial judges presiding in different courtrooms.  Independent realizations is sufficient for exchangeability. Datasets $D_1$ ($N = 51278$) and $D_2$ ($N = 51277)$ were constructed as random disjoint splits, which leaves exchangeability intact. 

% \section{Statistical Learning Forecasts in the Aggregate}
% \label{sec:aggregate}
\medskip
\noindent\textbf{Preliminary Data Analysis.}
The analysis required asymmetric costs for classification errors. For example, failing to correctly identify an offender who later is re-arrested for a violent crime has very different, and arguably more costly, consequences from failing to correctly identify an offender who is arrest free. The \emph{relative} cost of different classification errors should be built into a risk algorithm because otherwise, the forecasts will not properly represent the importance of different kinds of errors; the forecasts of risk can be seriously misleading. 

Arriving at target cost ratios requires input from stakeholders, and in practice, a consensus about relative costs usually is quickly reached. \citep{berk2019machine}. Generally, risk algorithms can be tuned so that target cost ratios are reasonably well approximated in empirical confusion tables. Suppose the target cost ratio for incorrectly classifying an offender, who is actually arrested for violent crime, as not being arrested at all was set at 5 to 1. Likewise, the target cost ratio for incorrectly classifying an offender, who is actually arrested for violent crime, as being arrested for a nonviolent crime was set at 2 to 1. Finally suppose the target cost ratio for incorrectly classifying an offender, who is actually arrested for nonviolent crime, as not being arrested at all also was set at 2 to 1. It is important to stress that such ratios represent subjective value preferences that are determined before a risk algorithm is applied. Different groups of stakeholders might well arrive at different cost ratios leading to different risk results. The target cost ratios just specified are reasonable for illustrative purposes.

Algorithmic training was undertaken with the training data $D_1$. Stochastic gradient boosting for a multinomial outcome was applied using the procedure \textit{XGBoost} in R \citep{chen2016xgboost}. Weighting was introduced to arrive empirically at sufficiently good approximations of the target cost ratios. Table~\ref{tab:confusion-table} is the resulting confusion table constructed for the $D_2$ data, here serving as the test dataset. 

\begin{table}[htp]
\caption{Out-of-sample data confusion table (based on $D_2$) for three outcome classes: no arrest, arrest for a nonviolent crime, arrest for a violent crime (N=51,277) }
\begin{center}
\begin{tabular}{|c|c|c|c|c|}
\hline
\diagbox[width=10em]{Actual}{Predicted} & No Arrest & Nonviolence & Violence & Classification Error\\
% Actual & Predict  & Predict & Predict & Classification \\
%  &  \\
\hline\hline
 No arrest & 18661 & 8120 & 3753  & 0.39 \\
 Nonviolence & 3617 & 10274 & 2410 & 0.37 \\
 Violence & 682 & 1009 & 2751 &  0.39 \\
 \hline
 Forecasting Error & 0.19 & 0.47 & 0.69 & \\
  \hline
\end{tabular}
\end{center}
\label{tab:confusion-table}
\end{table}

Entries in Table~\ref{tab:confusion-table} are the case counts. For example, 18,661 subjects out of 51,277 are \emph{correctly} predicted to be not arrested.  Similarly, 8,120 subjects who are not arrested are \emph{incorrectly} predicted to be arrested for a non-violent crime. Classification error and forecasting error are computed as follows. 

Classification error for a particular outcome, say no arrest, is the proportion of subjects \emph{erroneously} predicted to be arrested (for either non-violent/violent crime) among those who are not arrested. Mathematically, 
\begin{equation}\label{eq:classification-error}
\mbox{Classification Error (for no arrest)} := \frac{\sum_{i\in D_2} \mathbbm{1}\{\widehat{Y}_i \neq Y_i, Y_i = \mbox{no arrest}\}}{\sum_{i\in D_2} \mathbbm{1}\{Y_i = \mbox{no arrest}\}}.
\end{equation}
Here, $Y_i$ is the true outcome for subject $i$ in $D_2$ data, and $\widehat{Y}_i$ is the forecasted outcome from the trained classifier (i.e., the outcome with the highest estimated probability). The classification error in~\eqref{eq:classification-error} can be seen as an estimator of $\mathbb{P}(\widehat{Y} \neq Y|Y = \mbox{no arrest})$. From Table~\ref{tab:confusion-table}, the classification error for no arrest is estimated to be $0.39$. 

The forecasting error for a particular outcome, say no arrest, is the proportion of subjects \emph{erroneously} predicted to be not arrested among the subjects that are predicted to be not arrested. Mathematically,
\begin{equation}\label{eq:forecasting-error}
\mbox{Forecasting Error (for no arrest)} := \frac{\sum_{i\in D_2} \mathbbm{1}\{\widehat{Y}_i \neq Y_i, \widehat{Y}_i = \mbox{no arrest}\}}{\sum_{i\in D_2} \mathbbm{1}\{\widehat{Y}_i = \mbox{no arrest}\}}.      
\end{equation}
The notation $Y_i$ and $\widehat{Y}_i$ is unchanged. The forecasting error~\eqref{eq:forecasting-error} can be seen as an estimator of $\mathbb{P}(\widehat{Y} \neq Y|\widehat{Y} = \mbox{no arrest})$. From Table~\ref{tab:confusion-table}, the estimate for no arrest is $0.19$. 

% {\color{red}Given that we observe the forecast $\widehat{Y}$ in practice, it is more important to understand the forecasting error than the classification error: for how many offenders is the highest probability outcome the wrong outcome? Further, there exists cases where the highest probability forecast is not reasonable. For instance, if the estimated probabilities are $0.45$ for no arrest, $0.42$ for arrest with non-violent crime, and $0.13$ for arrest with violent crime, then the closeness of $0.45$ and $0.42$ should inform the analyst to report both no arrest and arrest with non-violent crime as ``best'' forecast instead of a single one. Some of the forecasting errors are made because of this and hence, shows the inadequacy of the confusion table.}

Classification error is important as a gauge of overall algorithmic performance. Hypothetically, when the true outcome class is known, how often is it incorrectly classified? One can see on the right margin of the table that all three outcome classes were misclassified roughly 40\% the time.

When real risk assessments are undertaken, attention shifts to forecasting error. The true outcome class is not known. Rather, a forecasted $\widehat{Y}$ is observed and used to help inform real decisions with real consequences. For how many subjects is the highest probability outcome the wrong outcome?

Forecasting error, reported at the bottom of the the table, ranged from approximately a fifth of the cases (i.e., 0.19) to about two-thirds of the cases (i.e., 0.69). The latter was substantially inflated by a large number classification errors caused by the high relative costs, specified by stakeholders, of erroneously classifying offenders who were actually violence prone, as posing no threat to public safety. A lower cost ratio would reduce the number of inflated risk estimates, but would substantially increase the number of violence prone offenders being overlooked. There are often tradeoffs in confusion tables between classification errors and forecasting errors even though they condition on different things.\footnote
{
Table~\ref{tab:confusion-table} shows that the stakeholder, target cost ratios were well approximated in the $D_2$ data. For example, the target 5 to 1 ratio when offenders who were re-arrested for a violent crime were classified as not re-arrested had an empirical cost ratio of about 5 to 1 (i.e. $3753/682 = 5.5$). In other words, one violence-prone offender misclassified as a risk-free was ``worth'' approximately five risk-free offenders misclassified as violence prone. It is very difficult to exactly reproduce the target cost ratios because proper confusion tables are constructed from test data, not training data. 
}

The likelihood of forecasting errors can be increased when the highest probability forecast is insufficiently definitive. For instance, if for a given offender the estimated probabilities are $0.45$ for no arrest, $0.42$ for an arrest alleging non-violent crime, and $0.13$ for an arrest alleging violent crime, the close proximity of $0.45$ and $0.42$ can imply that the classifier is unable to make an authoritative distinction between the two outcomes. Yet, for a Bayes classifier, the highest probability class automatically becomes the projected outcome. We show later that with conformal prediction sets, it often will be preferable to forecast more than one outcome class to more properly represent forecasting performance. 

Risk assessment results are usefully judged in comparison to current practice. The marginal distribution of the outcome classes represents that practice. A sentence of probation is determined by a sitting judge, and all of the offenders in this study were placed on probation; none were incarcerated. No quantitative forecasts of risk were used, but presumably the judge would not mandate probation for convicted offenders thought to be a public safety risk.\footnote
{
A smaller number of offenders were incarcerated, but they are not part of our study because there is no conduct on probation to measure.
}  

Suppose for the moment that implicitly sentencing judges makes decisions by the equivalent of a Bayes classifier. According to the empirical marginal distribution of the outcome classes, the ``default'' forecasts would be no subsequent arrest. In fact, 59\% of the offenders are not arrested. But 41\% of the offenders we were. The 41\% can be seen as erroneous forecasts.

With the available data analyzed using stochastic gradient boosting, Table~\ref{tab:confusion-table} shows that if no arrest is predicted, 19\% of the cases would be incorrectly forecasted. This is a meaningful improvement in forecasting accuracy compared to 41\%. It is bolstered by the cost ratios making failures to correctly classify new arrests, especially for violent crimes, relatively costly. The boosting algorithm works hard to avoid such errors; the counts of 3617 and 682 are smaller than they would have been if their relative costs were reduced. 

There are smaller improvements for the other two outcome classes had they been chosen by some other implicit forecasting method. If an arrest for a nonviolent crime is forecasted from the marginal distribution of the outcome classes, it would be wrong 68\% of the time. The same forecast from the risk algorithm would be wrong 47\% of the time. If an arrest for a violent crime is forecasted from the marginal distribution of the outcome classes, it would be wrong 91\% of the time. The same forecast from the risk algorithm would be wrong 69\% of the time. These are meaningful improvements in $D_2$ data that might be improved substantially with different cost ratio weighting. In short, by these measures of performance, forecasting error is reduced a nontrivial amount.\footnote
{
The random, disjoint split into $D_1$ and $D_2$ introduces additional uncertainty into the results. To get sense, we repeated the analysis several times, each with a random reconstruction of $D_1$ and $D_2$. There were noticeable changes in the results that typically improved the aggregate assessments of performance a bit. But, for purposes of this paper, they made no material difference.
} 

%\begin{center}
%\includegraphics[scale=.3]{Importance.pdf}
%\caption{$D_1$ In-Sample Predictor Relative %Contribution to the Gradient Boosting Fit}
%\label{fig:importance}
%\end{center}
%\end{figure}

 Many stakeholders are more confident about statistical learning forecasts when the dominant predictors comport well with expectations.
 \emph{XGBoost} provides a measure of each predictor's ``importance.'' Importance is measured as the average contribution to the fit for each tree used by the boosting algorithm in mean squared error units, then standardized so that the importance measure over all predictors sums to 100\%. This is an \emph{in-sample} approach distantly related to forecasting performance. There are more appropriate measures in forecasting settings,  available with other algorithms \citep[Section 10]{breiman2001random}. 

Four predictors of the seventeen are responsible for about half the fit: age, the number of prior jail sentences, the recency of latest prior arrest for a violent crime, and the earliest age at which an offender was arrested. Other analyses using partial dependence plots showed that the risks increase for younger offenders who had a larger number of prior jail sentences, who had an arrest for a violent crime more recently, and who began their criminal activities at an earlier age. No causal claims are implied because an algorithm is not a model, but none of the associations are a surprise \citep{berk2009role,berk2019machine}. 

 Without discussing at some length the criminal justice setting and the policy issues at stake, it is difficult to address whether the risk algorithm is performing well enough. Much of that discussion depends on what the future holds when the proposal to revise the risk algorithm proceeds. As the Black Lives Matter agenda illustrates well, the political environment is extremely volatile. Perceived fairness and transparency will matter as well. Still, one can properly claim that the algorithm's output is similar to the output of other algorithmic risk assessments in criminal justice for which the proper yardstick is current practice, not perfection \citep{berk2019machine}. 
 
But there is more. A risk algorithm with acceptable performance overall can be used to make forecasts for individual offenders. So far, our empirical results are silent on such matters. Yet as noted above, probation forecasts can have life-changing consequences. We will return to properties of classifier forecasts and confusion tables that apply to individual offenders after considering conformal prediction. Conventional classifiers and confusion tables when used to make forecasts about individuals can perform rather differently from conformal methods.  

\section{Some Background for Conformal Prediction}
\label{sec:nested-conformal-prediction}

Statistical learning classifiers commonly are used to make predictions about single cases. For ease of exposition, we consider for the moment binary categorical response variables. Illustrations include whether a spot on a lung X-ray is a cancer precursor~\citep{yan2018weakly}, whether a mortgage applicant will repay the loan~\citep{chen2020predicting}, or whether a particular hurricane will make landfall~\citep{alemany2019predicting}. One needs procedures that properly determine the precision and uncertainty of such forecasts.

Using the example of a hurricane landfall, the four possible prediction sets are \{Yes\}, \{No\}, \{Yes, No\} and $\emptyset$ (i.e., the empty set). The first could be a very informative if one can claim it is the true outcome class with a probability such as 0.90. Evacuation orders might follow. The second could be very informative if one can claim it is the true outcome class with a similarly high probability. There would be no evacuation orders. The third is an exercise in statistical humility; the forecasting procedure is unable to chose a single outcome class. This too would be important to know, perhaps making other information more important. The final prediction set can mean that the hurricane is atypical, and that the existing data do not apply. But this too may contain useful information. At a time when the consequences of climate change are rapidly materializing, it would be important to consider if a particular weather event is unprecedented.

The inferential task in these examples can be formulated in an unconditional manner; one only asks for guarantee \emph{on average} over all the forecasts. No guarantee is provided for a particular configuration of prediction values. Returning to a criminal justice setting, consider the task of predicting in a binary fashion whether an individual on probation is arrested for a violent crime based on the predictor vector $X = (X_1, X_2, X_3)$, with $X_1$ as age, $X_2$ as gender, and $X_3$ as the number of prior arrests. A prediction set $\widehat{C}(x)$ is said to have an \emph{unconditional} guarantee if for any future offender with predictor vector $X_{\mathrm{f}}$ and the true response $Y_{\mathrm{f}}\in\{0, 1\}$ one has
\begin{equation}
\label{eq:unconditional-guarantee}
\mathbb{P}\left(Y_{\mathrm{f}} \in \widehat{C}(X_{\mathrm{f}})\right) \ge 1 - \alpha,
\end{equation}
where $1 - \alpha$ is the specified probability such as $0.90$.\footnote
{
An arrest is coded as ``1,'' and an absence of an arrest is coded ``0.''
} 
In words,~\eqref{eq:unconditional-guarantee} conveys that if one constructs a $\widehat{C}(\cdot)$ for each of 100 future offenders, about $(1-\alpha)$ of these subjects will have their true response $Y_{\mathrm{f}}$ within their corresponding prediction set $\widehat{C}(X_{\mathrm{f}})$. This is the approach taken by conformal inference, whether the outcome is numeric or categorical, and it has strong finite sample guarantees~\citep{vovk2005algorithmic}. 

Forecasts can also be made in a conditional manner. For example, the distribution for whether there is a re-arrest for a violent crime can differ dramatically between male and female offenders and between younger and older offenders. In contrast to~\eqref{eq:unconditional-guarantee}, a prediction set $\widehat{C}(x)$ is said to have a conditional guarantee if for any future offender with predictor vector $X_{\mathrm{f}} = x$ and the true response $Y_{\mathrm{f}}\in\{0,1\}$, one has
\begin{equation}\label{eq:conditional-guarantee}
\mathbb{P}\left(Y_{\mathrm{f}}\in \widehat{C}(X_{\mathrm{f}})\,\big|\,X_{\mathrm{f}} = x\right) \ge 1 - \alpha\quad\mbox{for all}\quad x,
\end{equation}
where, again, $1-\alpha$ is the required probability such as $0.90$. In comparison to~\eqref{eq:unconditional-guarantee},~\eqref{eq:conditional-guarantee} provides the guarantee for any specific configuration of predictor values, and as before, whatever the forecast happens to be. %No additional calculations are required. The existing computations are used in a somewhat different manner. 

To briefly consider the practical importance of the conditional guarantee~\eqref{eq:conditional-guarantee}, suppose $x = (26, \text{Male}, 5)$; the individual is a male of age 26 with 5 prior arrests. Then~\eqref{eq:conditional-guarantee} conveys that if one provides a $\widehat{C}(\cdot)$ for each of 100 future male offenders of age 26 with 5 priors, about $(1-\alpha)$ of them will have their true response $Y_{\mathrm{f}}$ lie within a particular prediction set $\widehat{C}(26, \text{Male}, 5)$. Although each offender is processed one at a time, there is a distribution on the responses even after fixing the predictors to be $(26, \text{Male}, 5)$. Nothing specific is offered about a single offender. Still, one has a guarantee for sub-groups of offenders with the same configuration of predictor values.

Although unconditional guarantee~\eqref{eq:unconditional-guarantee} can be provided for finite samples as shown in~\cite{vovk2005algorithmic,lei2018distribution}, \cite{foygel2019limits} proved that the conditional guarantee~\eqref{eq:conditional-guarantee} is in general impossible to attain in finite samples. It is, however, possible to provide confidence sets satisfying the unconditional guarantee~\eqref{eq:unconditional-guarantee} in finite samples and to satisfy~\eqref{eq:conditional-guarantee} when the sample size is large enough. These results can be achieved using the \emph{nested} conformal approach advanced in~\citet[Appendix D]{gupta2019nested}; see also~\cite{izbicki2019flexible}. This approach is applicable for categorical and numeric responses. There are many settings in which conditional inference is preferable. In our application, it helps to insure that ``similarly situated'' offenders are treated similarly. 

Because of our policy setting, we emphasize nested conformal prediction sets for categorical response variables. \citet[Proposition 1]{gupta2019nested} prove the formal properties for the nested approach but their discussion focuses largely on numeric outcomes. Although the theory carries over, there are a number of practical considerations for categorical outcomes that need to be unpacked. Because conformal inference is relatively new to many statisticians, we include a substantial didactic material. Readers already familiar with nested conformal inference might wish to skim to Section~\ref{sec:categorical-Y}.

% For context, the next section of the paper briefly considers conditional nested prediction regions when the response variable is numeric. Section 5 introduces the rationale for nested conformal scores when the response variable is categorical. Section 6 addresses how such scores are used to construct nested conformal prediction sets appropriate for statistical inference. In section 7, a statistical learning classifier is applied to the data on probationers from which accuracy measures are constructed and compared.  Section 8 covers the lessons learned and offers more general conclusions.

\medskip

\noindent\textbf{Nested Conformal Prediction Regions for a Numeric $Y$.}
For a numeric/continuous $Y$, one has prediction intervals. For a categorical $Y$, one has prediction sets. We build, in particular, on prediction sets that are nested. For didactic purposes, however, we begin with a very brief discussion of nested prediction regions for a numeric $Y$. Details for categorical $Y$ are given in Section~\ref{sec:categorical-Y}.

Suppose we have exchangeable data $(X_i, Y_i), 1\le i\le n$ from a distribution $P,$ and the aim is to construct a region $\widehat{C}(x)$ such that if  a new case $(X_{n+1}, Y_{n+1})$ is exchangeable with $(X_i, Y_i), 1\le i\le n$, $\mathbb{P}(Y_{n+1}\in\widehat{C}(X_{n+1})) \ge 1 - \alpha.$ Intutively, smaller prediction regions for $\widehat{C}(x)$ are preferred; they are more precise. But the best prediction region depends on the \emph{true} conditional density $p(y|X = x) = p(y|x)$ and is given by
\begin{equation}\label{eq:oracle-prediction}
\widehat{C}^{\mathrm{oracle}}(x) ~:=~ \left\{y\,\big|\,p(y|x) \ge t^{\mathrm{oracle}}(\alpha, x)\right\},
\end{equation}
where a threshold $t^{\mathrm{oracle}}(\alpha; x)$ is chosen to be the largest possible value such that~\eqref{eq:oracle-prediction} holds for $\widehat{C}^{\mathrm{oracle}}(x)$. A higher threshold implies a smaller prediction region, but the best prediction region cannot be computed in practice because the \emph{true} conditional density $p(y|x)$ and hence, $t^{\mathrm{oracle}}(\alpha, x)$, are unknown.

\begin{figure}[htbp]
\begin{center}
\includegraphics[width=\textwidth, height=3.2in]{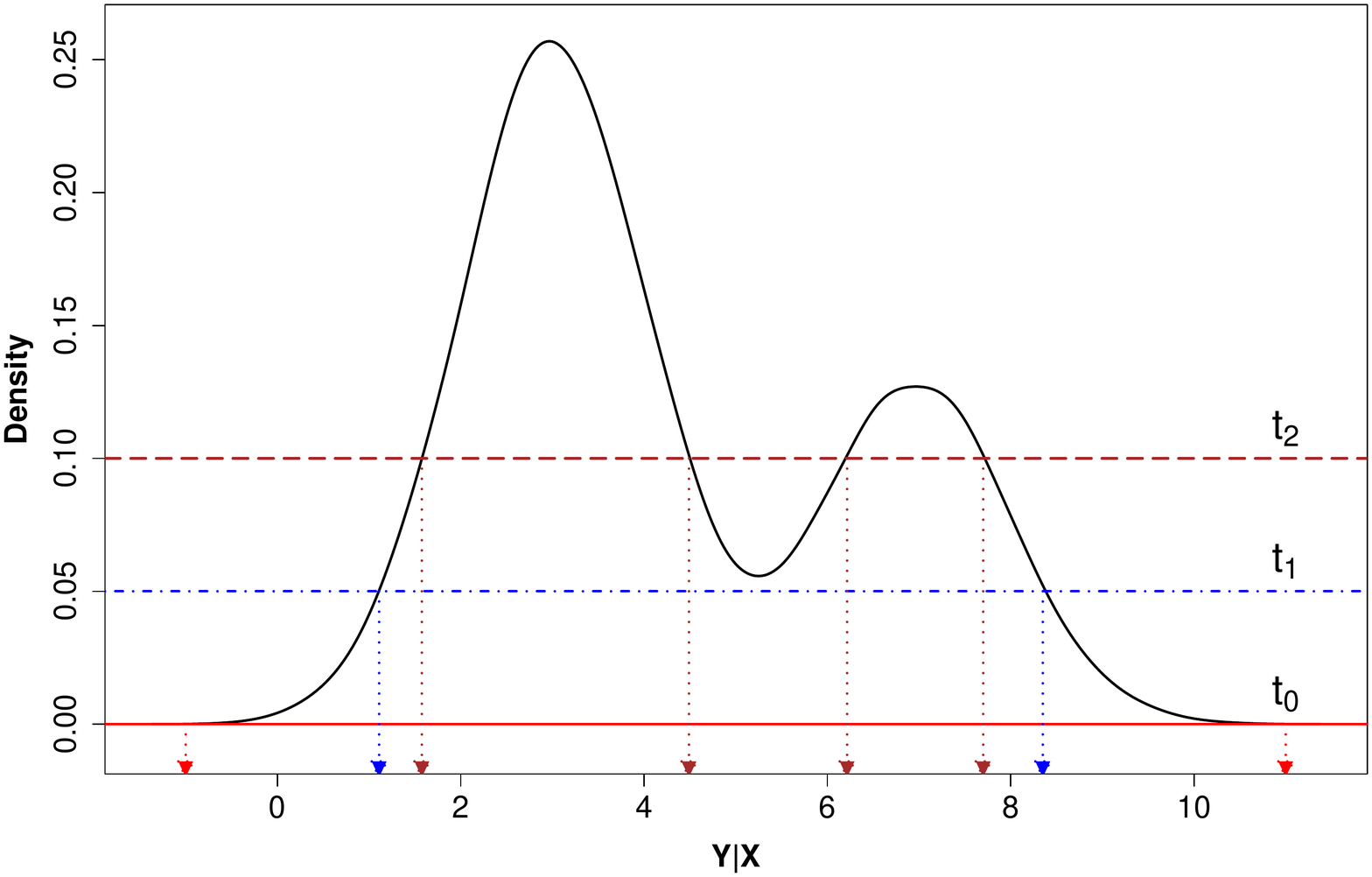}
\caption{Nested Prediction Regions for a Numeric $Y|\textbf{X}$ = $\textbf{x}$ with Three Illustrative Thresholds}
\label{fig:numeric}
\end{center}
\end{figure}
Estimation is addressed shortly. For now, we illustrate properties of the  ``oracle'' prediction region with a bimodal density. Figure~\ref{fig:numeric} is a plot of the true density $p(y|x)$ (for a specific $x$) with three different thresholds. For the threshold at $t_0 = 0.0$ (solid line), the region $\{y|p(y|x) \ge t_0\} = \{y|p(y|x) \ge 0\}$ is the set of all possible $y$'s between the two red arrowheads. The coverage probability for this set is $1$. Increasing the threshold to $t_1$ (in dash-dotted line) removes some response values and concurrently, the coverage guarantee is reduced. The set of possible $y$'s is represented by the values between the blue arrowheads, and the area under the density above $t_1$ is smaller than for $t_0$. The same reasoning applies of $t_2$ (in dashed line). The set of possible $y$'s corresponds to the values between the two brown arrowheads on the left combined with the values between the two brown arrowheads on the right.

Increasing the threshold $t$ to the top of the peak on the left, leaves no value of $y$ that satisfies $p(y|x) \ge t$, and hence, the set $\{y|\,p(y|x) \ge t\}$ is empty. The coverage probability for this set is $0$. Changing the threshold $t$ from $0$ to the top of the peak, the coverage probability for the set $\{y|\,p(y|x) \ge t\}$ drops smoothly from $1$ to $0$ and equals the required probability of $1 - \alpha$ at some threshold. That special threshold is $t^{\mathrm{oracle}}(\alpha, x)$, and is \emph{the best prediction set} given $1-\alpha$. 

Estimation of best prediction set is central to all that follows. Consider an estimate $\widehat{p}(y|X = x) = \widehat{p}(y|x)$ for the true conditional density $p(y|x)$. Using this estimator, we obtain an estimate $\widehat{t}(\alpha, x)$ of the oracle threshold $t^{\mathrm{oracle}}(\alpha, x)$. 

The ``naive'' set 
\begin{equation}\label{eq:naive-prediction}
\widehat{C}^{\mathrm{naive}}_{\alpha}(x) := \{y:\,\widehat{p}(y|x) \ge \widehat{t}(\alpha, x)\},
\end{equation}
cannot have any finite sample guarantee because $\widehat{p}(y|x)$ is an estimate with unknown arbitrary accuracy. But for any $\gamma\in[0, 1]$, and $\mathbb{P}(Y_{\mathrm{f}} \in \widehat{C}^{\mathrm{naive}}_{\gamma}(X_{\mathrm{f}}))$ not greater than $1 - \gamma$, there is a $1 - f(\gamma)$, for some function $f(\cdot)$ that provides for the desired finite sample guarantees. By making $\gamma = f^{-1}(\alpha)$,
\[
\mathbb{P}\left(Y_{\mathrm{f}} \in \widehat{C}^{\mathrm{naive}}_{\gamma}(X_{\mathrm{f}})\right) = 1 - f(\gamma) = 1 - f(f^{-1}(\alpha)) = 1 - \alpha.
\]
One has a form of calibration. Because $\widehat{C}^{\mathrm{naive}}_{\alpha}(\cdot)$ does not have the desired coverage, $\alpha$ is altered so that the coverage becomes the original $1 - \alpha$. However, $f(\cdot)$ is unknown and hence, conformal prediction is used to calibrate the coverage. The mathematical details of $\widehat{t}(\alpha, x)$, as well as calibration for a numeric $Y$, is beyond the scope of this paper. Interested readers should consult~\citet[Appendix D]{gupta2019nested}.

With this background in place, the next section considers categorical outcome classes in detail. 

\section[Nested for Categorical]{Nested Conformal Prediction Sets for Categorical $Y$}\label{sec:categorical-Y} 

For a categorical $Y$, the basic ideas are similar, but the details differ substantially because the thresholds are applied to a probability distribution, not a probability density, and response classes are on the horizontal axis. Just as in Figure~\ref{fig:numeric}, the threshold can increased smoothly, but the probability above the threshold goes from 1 to 0 in abrupt steps. Included y-classes change abruptly too. For a given value of $\alpha$, the goal is, as before, to determine the best prediction set for $1 - \alpha$ using calibration. The abrupt transitions complicate the process.

To be consistent with the probation application, three outcome classes are used. The reasoning to follow applies equally well when there are two outcome classes or more than three. However, with two classes, important ideas are lost, and with more than three classes, foundational ideas can be obscured by additional details.

Consider first an  ``oracle'' probability distribution for a categorical $Y$ and a fixed $x$ with three outcome classes coded as $0, 1, 2$. Their probabilities are
\[
p(0|x) = 0.43,\quad p(1|x) = 0.35,\quad\mbox{and}\quad p(2|x) = 0.22,
\]
displayed in Figure~\ref{fig:pmf-plot} by three vertical lines in gray. Three illustrative thresholds are shown as horizontal lines in red, blue or green. The vertical axis is in probability units.

\begin{figure}[htbp]
\begin{center}
\includegraphics[width = \textwidth, height=3in]{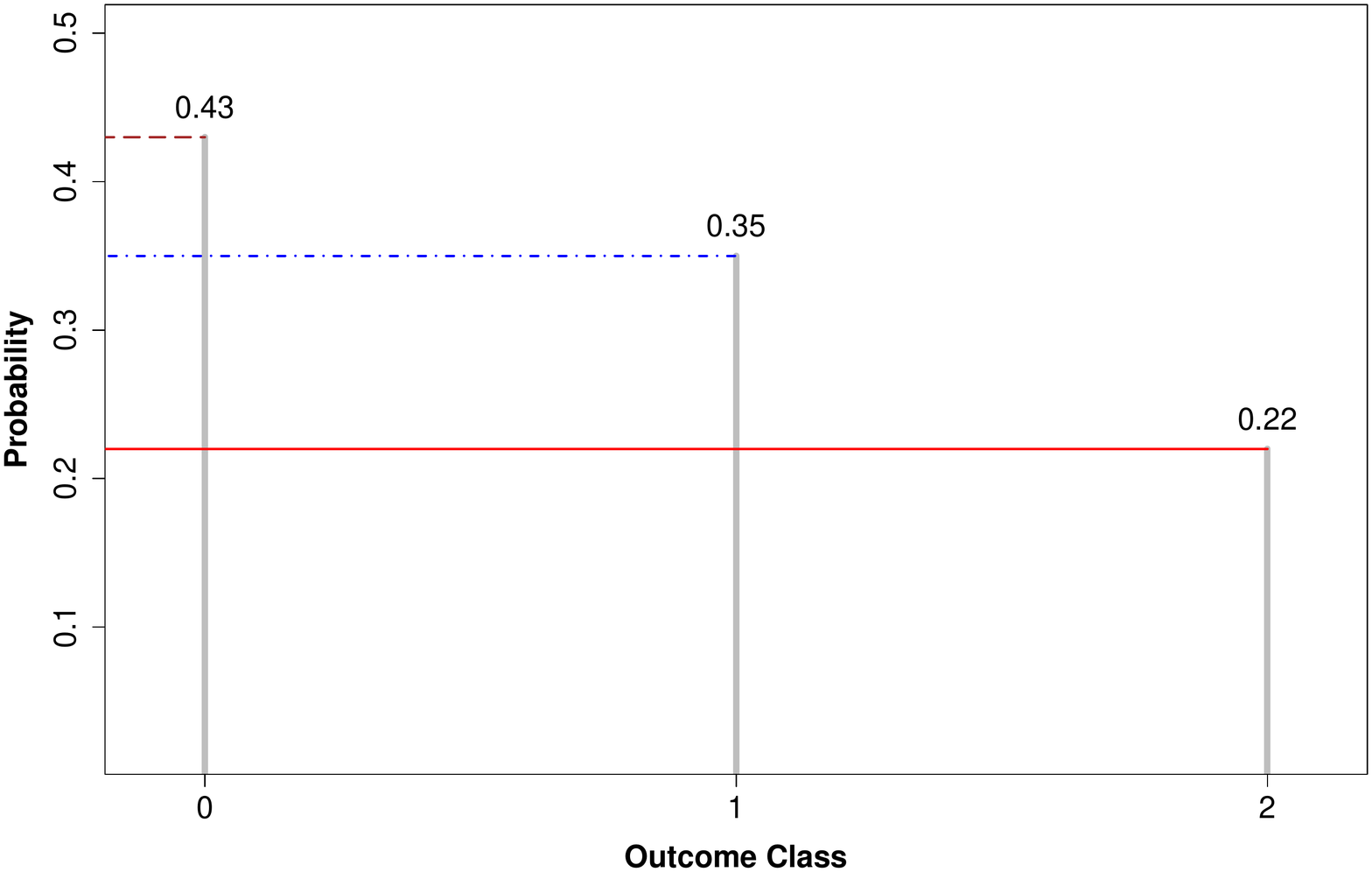}
\caption{For a single case: an oracle conditional probability distribution $Y|\textbf{X}$ = $\textbf{x}$ for a response variable with three classes having three probabilities 0.43, 0.35, 0.22 and three transitional thresholds in brown, blue, and red, respectively.}
\label{fig:pmf-plot}
\end{center}
\end{figure}

For the full oracle prediction set $\{y\in\{0, 1, 2\}:\, p(y|x) \ge t^{\mathrm{oracle}}(\alpha, x)\}$, where the threshold $t^{\mathrm{oracle}}(\alpha, x)$ is chosen so that the region has $1 - \alpha$ coverage. This closely parallels the earlier discussion for  a numeric $Y$. As the threshold is smoothly increased, the probability above the threshold is the sum of the class probabilities for the remaining classes. That probability at the red threshold and above is $0.22 + 0.35 + 0.43 = 1.0.$ The coverage for \{0,1,2\} is 1.0. Should the specified $\alpha = 0.30$, making $0.70$ the appropriate coverage, the actual coverage very conservative, and all three outcome classes are in the prediction set. Clearly, this is an unsatisfying  result. The probability at or above the blue threshold is $0.35 + 0.43 = 0.78$. Coverage for \{0,1\} is 0.78. This prediction set is smaller, but the actual coverage is still a bit conservative. Finally, the probability at or above the green threshold is $0.43$. Coverage for \{0\} is 0.43, which falls below 0.70, and contravenes the coverage target. 

Table~\ref{tab:threshold-calculation-conformal} provides an arithmetic translation of these relationships. Consider in column 1 all illustrative thresholds increasing from a value of 0.0 to include but not exceed a value 0.22. For each, all outcome classes are included with a coverage of 1.0, which is larger than the threshold value. For all such thresholds, one has a valid but, as before, a very conservative and uninteresting prediction set.

For an arbitrarily small $\delta$, there are new thresholds beginning at $\delta$ above 0.22 and ending at, but not exceeding, .35. None of these thresholds are above $0.43 + 0.35=.78$. The included outcome classes are \{0,1\}; \{2\} no longer qualifies. (See Figure~\ref{fig:pmf-plot}.) 

From thresholds from $\delta$ above 0.35 up to but not exceeding 0.43, the single class \{0\} remains; the other two classes are eliminated. (See again Figure~\ref{fig:pmf-plot}.) Although the remaining probability is still greater than the threshold values, the coverage falls well below the desired 0.70. 

For threshold values between $0.43 + \delta$ and $1.0$, the sum of the class probabilities is equal to 0.0 and, therefore, smaller than all of the threshold values. There are no outcome classes left, and one has the empty prediction set $\emptyset$. Clearly, this result useless, at least from a policy perspective. 

\begin{table}[htp]
\caption{Computing $t^{\mathrm{oracle}}(\alpha, x)$ for $\alpha = 0.30$ to obtain a prediction set for a categorical $Y$ with three classes $\{0,1,2\}$: $p(0|x)=0.43$, $p(1|x)=0.35$, and $p(2|x)=0.22$. $\delta > 0$ is an arbitrarily small number. The rows where the threshold changes the sum of probabilities are shaded gray.}
\begin{center}
\begin{tabular}{|c|c|c|}
\hline 
Threshold $t$ & $\{y:p(y|x) \ge t\}$ & Sum of Class Probabilities  \\
% $t$ & Sets & in Region $ > t_{\gamma}$ \\
\hline\hline
0.0 &\{0,1,2\} & 0.43 + 0.35 + 0.22 = 1.0 \\
0.10 & \{0,1,2\} & 0.43 + 0.35 + 0.22 = 1.0 \\
0.20 & \{0,1,2\} & 0.43 + 0.35 + 0.22 = 1.0 \\
\rowcolor{Gray}
0.22 & \{0,1,2\} & 0.43 + 0.35 + 0.22 = 1.0 \\
\rowcolor{Gray}
0.22 + $\delta$ & \{0,1\} & 0.43 + 0.35 = 0.78 \\
0.30 & \{0,1\} & 0.43 + 0.35 = 0.78 \\
\rowcolor{Gray}
0.35 & \{0,1\} & 0.43 + 0.35 = 0.78 \\
\rowcolor{Gray}
0.35 + $\delta$ & \{0\} & 0.43\\
0.40 & \{0\} & 0.43 \\
\rowcolor{Gray}
0.43 & \{0\} & 0.43 \\
\rowcolor{Gray}
0.43 + $\delta$ & $\emptyset$ & 0.0 \\
0.5 & $\emptyset$ & 0.0\\
1.0 & $\emptyset$ & 0.0\\
\hline %\hline
\end{tabular}
\end{center}
\label{tab:threshold-calculation-conformal}
\end{table}

From this reasoning and, say, $\alpha = 0.30$, one can choose any threshold with values between $[0, 0.35]$ and guarantee coverage of at least $0.7$. By convention and in the spirit the discussion for a numeric $Y$, we take the largest threshold satisfying the coverage, and this yields $t^{\mathrm{oracle}}(0.3, x) = 0.35$. 
\medskip

\noindent\textbf{Estimation.}
In practice, an estimator for the oracle probability distribution is required, which means that the class probabilities and thresholds need to be estimated as well. A threshold now is denoted by $\widehat{t}(\gamma,x)$, where $\gamma$ can be seen as a provisional value of $\alpha$ such that $1-\gamma$ is a provisional coverage. The term ``provisional'' is employed because the impact of different $\gamma$-values will be examined. Using reasoning much like that for the oracle setting, a given estimator $\widehat{p}(y|x)$ for $y\in\{0, 1, 2\}$, and any $\gamma\in[0, 1]$, one can obtain $\widehat{t}(\gamma, x)$ and its ``naive'' prediction set
\[
\widehat{C}_{\gamma}^{\mathrm{naive}}(x) ~:=~ \{y\in\{0,1,2\}:\,\widehat{p}(y|x) \ge \widehat{t}(\gamma, x)\}.
\]

In words, a naive prediction set for a specified $\gamma$ includes one or more outcome classes, here from \{0,1,2\}, such that their estimated probability is equal to or larger than the provisional $\hat{t}(\gamma,x)$ for the provisional $\gamma$.

As noted in Section~\ref{sec:nested-conformal-prediction}, the naive prediction set does not have a coverage guarantee because $\widehat{p}(y|x)$ is only an estimate. It is possible to do better using calibration. If $\widehat{p}(y|x) \equiv p(y|x)$, then $\widehat{C}_{\gamma}^{\mathrm{naive}}(x)$ has the $1 - \gamma$ coverage guarantee. One implication that a researcher's choice of classifier can help insofar as the 
$p(y|x)$ is better approximated. Moreover, a more accurate estimator $\widehat{p}(\cdot|\cdot)$ can reduce the $\gamma$ calibration burdens and facilitate a closer representation of the conditional validity~\eqref{eq:conditional-guarantee}. 

\medskip

\noindent\textbf{Calibration using Conformal Prediction.}
Rather than trying to alter $\widehat{p}(\cdot|\cdot)$, calibration proceeds directly to goal of the best prediction set, not indirectly by trying to make $\widehat{p}(y|x) \equiv p(y|x)$. The value of $\gamma$ for $\widehat{C}_{\gamma}^{\mathrm{naive}}(x)$ is altered such that finite sample $1 - \alpha$ coverage is achieved. Stated more formally, calibration seeks a value of $\gamma$, such that $\gamma = \gamma(\alpha)$, from which
\[
\mathbb{P}\left(Y_{\mathrm{f}} \in \widehat{C}_{\gamma(\alpha)}^{\mathrm{naive}}(X_{\mathrm{f}})\right) ~\ge~ 1 - \alpha.
\]
Because the probability of an event can be approximated by averaging over sample points, $\gamma(\alpha)$ can be effectively estimated as
\begin{equation}\label{eq:calibration}
\frac{1}{|D_2|}\sum_{i\in D_2}\mathbbm{1}\left\{Y_i \in \widehat{C}^{\mathrm{naive}}_{\gamma(\alpha)}(X_i)\right\} \ge 1 - \alpha.
\end{equation}
For an exact finite sample guarantee, the right hand side should be replaced by $(1 + 1/|D_2|)(1 - \alpha)$; see Lemma 2 of~\cite{romano2019conformalized} for a proof. This small increment in $(1-\alpha)$ is a requirement for valid inference.

\subsection*{Conformal Scores}
Conformal scores are the building blocks for calibration. For any $(x, y)$, a conformal score is defined as
\begin{equation}\label{eq:nested-conformal-score}
s(x, y) ~:=~ \sup\left\{\gamma\in[0,1]:\,y\in\widehat{C}^{\mathrm{naive}}_{\gamma}(x)\right\}.
\end{equation}
In words, $s(x, y)$ is the largest $\gamma\in[0, 1]$  such that the naive prediction set at $\hat{t}(\gamma, x)$ contains at least one element of $y$ corresponding to an actual outcome class. For example, the class might be no arrest while on probation. Illustrations are provided shortly in Table~\ref{tab:conformal_score}. 

There is nesting because $\widehat{C}^{\mathrm{naive}}_{\gamma}(x)$ is monotonically decreasing in $\gamma$, such that for $\gamma_1 \le \gamma_2$, 
\[
\widehat{C}^{\mathrm{naive}}_{\gamma_2}(x) ~\subseteq~ \widehat{C}^{\mathrm{naive}}_{\gamma_1}(x).
\]

Because the naive prediction sets are nested from $\widehat{C}_{\gamma = 0}^{\mathrm{naive}}(x) = \{0, 1, 2\}$ to $\widehat{C}^{\mathrm{naive}}_{\gamma = 1}(x) = \emptyset$, the threshold can be increased until the largest $\gamma$ is chosen such that $y\in\widehat{C}^{\mathrm{naive}}_{\gamma}(x)$ is true. The nestedness of the sets and the conformal score as defined in~\eqref{eq:nested-conformal-score} constitute the nested conformal framework in~\cite{gupta2019nested}.  

To help fix the ideas in~\eqref{eq:nested-conformal-score}, consider some new illustrative calculations for Table~\ref{tab:conformal_score} with the goal of computing a proper conformal score. The oracle probability distribution now is treated as a property of a single, realized, exchangeable case. This is fundamentally different from the earlier setting in which Table~\ref{tab:threshold-calculation-conformal} was placed, although some of the reasoning and calculations will be familiar. Attention centers on the role of $\widehat{t}(\gamma, x)$ for all $\gamma\in[0, 1]$. The outcome classes remain \{0,1,2\} and, as before, suppose $\widehat{p}(0|x) = 0.43$, $\widehat{p}(1|x) = 0.35,$ and $\widehat{p}(2|x) = 0.22$.

Starting at the top of Table~\ref{tab:conformal_score}, one has a $\widehat{t}(\gamma,x)$ with $\gamma = 0.0$. In this role, $\gamma$ is a place-holder for $\alpha$ determining the provisional prediction set coverage. Moving downward, the last row for which the sum of the three outcome class probabilities at least $1 - \gamma = 1.0$ determines an initial $\widehat{t}(\gamma,x)$, which applies to all $\gamma < 0.22$. Because $\delta$ can be arbitrarily small, $\widehat{t}(0, x) = 0.22$. 
Similarly, for any $\gamma\in[0, 0.22)$, $\widehat{t}(\gamma, x) = 0.22$, and the prediction set is the same. Much as for a numeric $Y$, higher thresholds are preferred, other things equal, because the prediction regions are smaller even though in this case the prediction set does not change for $\gamma\in[0, 0.22)$. 

When $\gamma = 0.22$, $1 - \gamma = 0.78$ (i.e. $.35 + .43)$. By the same reasoning, the last instance when the sum of probabilities in Table~\ref{tab:conformal_score} is at least $0.78$ is for $\widehat{t}(\gamma,x)$ with $\gamma < 0.57$. One has for $\gamma\in[0.22, 0.57)$, a different $\widehat{t}(\gamma, x) = 0.35$, and prediction set is \{0,1\}. 

For $\gamma = 0.57$, $1 - \gamma = 0.43$, the last row for which the sum of probabilities in Table~\ref{tab:conformal_score} is at least $0.43$. Hence, $\widehat{t}(0.57, x) = 0.43$. A Similar argument implies that for $\gamma\in[0.57, 1)$, $\widehat{t}(\gamma, x) = 0.43$. Finally, $\widehat{t}(1, x) = 1$. 

For the calculations summarized in Table~\ref{tab:conformal_score}, one imagines fixing the threshold at the smallest outcome class probability and increasing the value of $\gamma$ until Equation~\ref{eq:nested-conformal-score} is satisfied. Then, one fixes the threshold at next highest outcome class probability and again, increases the value of $\gamma$ until Equation~\ref{eq:nested-conformal-score} is satisfied. This process is repeated until the outcome classes are exhausted. 

\begin{table}[htp]
\caption{Construction of conformal score for categorical $Y|\textbf{X}$ = $\textbf{x}$ with three classes 0, 1, and 2: $\widehat{p}(0|x)=0.43$, $\widehat{p}(1|x)=0.35$, and $\widehat{p}(2|x)=0.22$. $\delta > 0$ is an arbitrarily small number. The rows where a change in $\gamma$ changes the prediction set are shaded gray.}
\begin{center}
\begin{tabular}{|c|c|c|}
\hline
$\gamma$ & $\widehat{t}(\gamma, x)$ & $\widehat{C}^{\mathrm{naive}}_{\gamma}(x) = \{y:\widehat{p}(y|x) \ge \widehat{t}(\gamma,x)\}$  \\
\hline\hline
0.0 & 0.22 & \{0, 1, 2\} \\
0.20 & 0.22 & \{0, 1, 2\} \\
\rowcolor{Gray}
$0.22 - \delta$ & 0.22 & \{0, 1, 2\}\\
\rowcolor{Gray}
0.22 & 0.35 & \{0, 1\} \\
0.30 & 0.35 & \{0, 1\}\\
\rowcolor{Gray}
$0.57 - \delta$ & 0.35 & \{0, 1\}\\
\rowcolor{Gray}
0.57 & 0.43 & \{0\}\\
0.7 & 0.43 & \{0\}\\
\rowcolor{Gray}
$1.0 -\delta$ & 0.43 & \{0\}\\
\rowcolor{Gray}
1.0 & 1.0 & $\emptyset$\\
\hline %\hline
\end{tabular}
\end{center}
\label{tab:conformal_score}
\end{table}

It follows that for Table~\ref{tab:conformal_score}, there are three conformal scores to be computed: $s(x, 0), s(x, 1),$ and $s(x, 2)$. Starting with $s(x, 0)$, one requires the largest $\gamma$ such that $0$ belongs in $\widehat{C}^{\mathrm{naive}}_{\gamma}(x)$. From Table~\ref{tab:conformal_score}, this is the set of all $\gamma$ such that $0\in\widehat{C}_{\gamma}^{\mathrm{naive}}(x)$ is $[0, 1)$, the supremum of which is $1$. Consequently, $s(x,0) = 1$. For $s(x, 1)$, the set of all $\gamma$ such that $1\in\widehat{C}_{\gamma}^{\mathrm{naive}}(x)$ is $[0, 0.57)$, the supremum of which is $0.57$. Consequently, $s(x, 1) = 0.57$. Finally, the set of all $\gamma$ such that $2\in\widehat{C}_{\gamma}^{\mathrm{naive}}(x)$ is $[0, 0.22)$, the supremum of which is $0.22$. Consequently, $s(x, 2) = 0.22$. In short, the three conformal scores, one for each outcome class from \{0,1,2\}, are
\[
s(x, 0) = 1,\quad s(x, 1) = 0.57,\quad\mbox{and}\quad s(x, 2) = 0.22.
\]
Such reasoning conveys why $s(x, y)$ is called a conformity score. For classifiers providing output probabilities, if $\widehat{y}$ is the chosen label for some $x$ (i.e. the fitted classs is the same as the observed class), $s(x, \widehat{y})$ will have the largest value among the outcome class conformal scores. Here, from the estimated probability distribution, $0$ is the most likely outcome and hence, $(x, 0)$ conforms most closely to a classifier's highest probability selection. It is appropriate, therefore, that $s(x, 0)$ is the largest conformal score. The second most probable $s(x, 1)$ conforms less well, and the third most probable $s(x, 2)$ conforms least well. 

More formally, let $\widehat{\pi}(0), \widehat{\pi}(1), \widehat{\pi}(2)\in\{0,1,2\}$ re-define the class labels such that
\begin{equation}\label{eq:ordering-probabilities}
\widehat{p}(\widehat{\pi}(0)|x) ~>~ \widehat{p}(\widehat{\pi}(1)|x) ~>~ \widehat{p}(\widehat{\pi}(2)|x).
\end{equation}
The new outcome classes are arranged in order from the largest estimated probability to the smallest and re-labeled respectively as $\widehat{\pi}(0), \widehat{\pi}(1),$ and $\widehat{\pi}(2)$. For example, is an arrest for a nonviolent crime has the second largest probability, it is denoted by $\widehat{p}(\widehat{\pi}(1))$. Then, 
\begin{equation}\label{eq:conformal-scores-ordering}
    \begin{split}
        s(x, \widehat{\pi}(0)) ~&=~ 1,\\
        s(x, \widehat{\pi}(1)) ~&=~ \widehat{p}(\widehat{\pi}(1)|x) + \widehat{p}(\widehat{\pi}(2)|x),\\
        s(x, \widehat{\pi}(2)) ~&=~ \widehat{p}(\widehat{\pi}(2)|x).
    \end{split}    
\end{equation}
For a concrete understanding of the formulae~\eqref{eq:ordering-probabilities} and~\eqref{eq:conformal-scores-ordering}, consider an example case where the estimated probabilities of outcome classes are
\[
\widehat{p}(0|x) = 0.27,\quad \widehat{p}(1|x) = 0.54,\quad\mbox{and}\quad\widehat{p}(2|x) = 0.19.
\]
For this study subject, the class with largest probability is $1$, the class with second largest probability is $0$, and the class with third largest (or the smallest) probability is $2$. This means $\widehat{\pi}(0) = 1$, $\widehat{\pi}(1) = 0$, and $\widehat{\pi}(2) = 2$ in~\eqref{eq:ordering-probabilities}. Applying~\eqref{eq:conformal-scores-ordering} yields
\begin{equation}\label{eq:example-conformal-score}
\begin{split}
s(x,\widehat{\pi}(0)) &= s(x, 1) = 1,\\
s(x,\widehat{\pi}(1)) &= s(x, 0) = 0.27 + 0.19 = 0.46,\\
s(x,\widehat{\pi}(2)) &= s(x, 2) = 0.19.
\end{split}
\end{equation}
Note that computing the scores does not require knowledge of the observed outcome. In this example, if the observed outcome is $0$, then the conformal score for the observation $(x, 0)$ is $s(x, 0) = 0.46.$
% {\color{red} I don't follow these two sentences. What's the point?}

If any two of the probabilities in~\eqref{eq:ordering-probabilities} are equal, the formula for $s(x, y)$ is more complicated, but  one can always add a small amount of noise to all the probabilities so that no two probabilities are equal. This can be done without changing the predictions or prediction sets. 
% If for this offender the outcome class of an arrest for a violent crime had the largest estimated probability, $\widehat{p}(\widehat{\pi}(0)|x)$ would place that class first. If no arrest had the second largest probability, $\widehat{p}(\widehat{\pi}(1)|x)$ would place that class second. And if an arrest for a nonviolent crime had the smallest probability, $\widehat{p}(\hat{\pi}(2)|x)$ would place that class third.

% Then, $s(x, \widehat{\pi}(0)) = 1,$ $s(x, \widehat{\pi}(1)) = \widehat{p}(\widehat{\pi}(1)|x) + \widehat{p}(\widehat{\pi}(2)|x),$ and $s(x, \widehat{\pi}(2)) = \widehat{p}(\widehat{\pi}(2)|x)$.

% Note that, although $s(x, y)$ is the ``largest'' $\gamma$ such that $y\in \widehat{C}_{\gamma}^{\mathrm{naive}}(x)$, in the example above $y$ \emph{does not} belong to $\widehat{C}^{\mathrm{naive}}_{s(x,y)}(x)$. It misses very narrowly. For any $\gamma < s(x, y)$, $y\in \widehat{C}^{\mathrm{naive}}_{\gamma}(x)$, but when $\gamma = s(x, y)$, the inclusion does not hold. This property of conformal scores becomes important in calibration and the final prediction set. 

\subsection*{Calibration}

With one conformal score in hand for each $D_2$ case, calibration is undertaken by calculating an appropriate quantile. For all $i\in D_2$, and conformal scores $s(X_i, Y_i)$, $\widehat{\gamma}(\alpha)$ is defined by
\[
1 - \widehat{\gamma}(\alpha) ~:=~ \left(1 + \frac{1}{|D_2|}\right)(1 - \alpha)\mbox{th quantile of }1 - s(X_i, Y_i), i\in D_2. 
\]
The best prediction set is \
\begin{equation}\label{eq:nested-conformal-prediction-set}
\begin{split}
\widehat{C}_{\alpha}(x) ~&:=~ \left\{y:\, s(x, y) \ge \widehat{\gamma}(\alpha)\right\} = \left\{y:\,\widehat{p}(y|x) \ge \max_{\gamma < \widehat{\gamma}(\alpha)}\widehat{t}(\gamma, x)\right\}.
\end{split}
\end{equation}
Because $s(x, y)$ is a conformity score indicating how well $(x, y)$ conforms to the $D_2$ classifications, $1 - s(x, y)$ actually is a non-conformity score. The best prediction set $\widehat{C}_{\alpha}(x)$ is collection of all $y$ that conforms the best to the actual $D_2$ outcome classes within the required coverage. This is accomplished by only taking those $y$ for which $1 - s(x, y)$, the non-conformity score, is less than $1 - \widehat{\gamma}(\alpha)$.
The earlier discussion and the nested conformal method for categorical $Y$ is summarized by Algorithm~\ref{alg:nested-conformal-prediction}.
\begin{algorithm}[h]
    \caption{Nested conformal prediction for classification}
    \label{alg:nested-conformal-prediction}
    \SetAlgoLined
    \SetEndCharOfAlgoLine{}
    \KwIn{Data splits $D_1$ and $D_2$, coverage probability $1 - \alpha$.}
    \KwOut{A prediction set $\widehat{C}_{\alpha}(\cdot)$ such that $\mathbb{P}(Y_{\mathrm{f}}\in\widehat{C}(X_{\mathrm{f}})) \ge 1 - \alpha$.}
    Train a classifier $\widehat{p}(\cdot|\cdot)$ on $D_1$. This gives a probability distribution (estimator) for the outcomes from each $x$.\;
    % \hspace{0.1in} 
    Calculate the conformal scores $s(X_i, Y_i)$ for all $i\in D_2$ defined in~\eqref{eq:nested-conformal-score}.\;
    % \hspace{0.1in} 
    Compute the $(1 + 1/|D_2|)(1 - \alpha)$-th quantile of $1 - s(X_i, Y_i), i\in D_2$. Call this quantile $1 - \widehat{\gamma}(\alpha)$. This is essentially the $\lceil(|D_2| + 1)(1 - \alpha)\rceil$-th largest value in the sequence $1 - s(X_i, Y_i), i\in D_2$\;
    \Return the prediction set 
    \begin{equation}\label{eq:nested-predict-set-algo}
    \widehat{C}_{\alpha}(x) ~:=~ \left\{y:\,s(x, y)\ge\widehat{\gamma}(\alpha)\right\}.
    \end{equation}
\end{algorithm}
\begin{thm}\label{THM:SPLIT-NESTED-CONFORMAL}
If the data $(X_i, Y_i), i\in D_2$ are exchangeable, then the prediction set $\widehat{C}_{\alpha}(\cdot)$ obtained from Algorithm~\ref{alg:nested-conformal-prediction} satisfies
\[
\mathbb{P}\left(Y_{\mathrm{f}}\in\widehat{C}_{\alpha}(X_{\mathrm{f}})\right) \ge 1 - \alpha.
\]
\end{thm}
\begin{proof}
The result follows from Proposition 1 of~\cite{gupta2019nested}.
\end{proof}
\begin{rem}
The approach taken in Algorithm~\ref{alg:nested-conformal-prediction} is called the split conformal prediction method~\citep{papadopoulos2002inductive,lei2014distribution}. There are several, somewhat more involved, versions of conformal methods that do not require sample splitting, and can in principle make better use of the data: jackknife+, CV+, subsampling, and bootstrap. Even a brief discussion of these procedures would take us far afield, and excellent treatments are easily found  in~\cite{barber2019predictive}, \cite{gupta2019nested}, and~\cite{kim2020predictive}. 
% \color{red} (I stopped here. I will construct some new tables below before getting back to the writing.) 
\end{rem}
\medskip

\noindent\textbf{Applying  Algorithm~\ref{alg:nested-conformal-prediction} to Probation the Data.} Tables~\ref{tab:conformalscore} and \ref{tab:examples} provide illustrations of the key features of Algorithm 1. Later, connections to the classifier output in Table~\ref{tab:confusion-table} will be addressed.

\begin{table}[htp]
\caption{Examples of conformal scores based on the estimated class probabilities and the actual outcome class. The table shows $s(X_i, Y_i)$ for six random $i$'s in $D_2$.}
\begin{center}
\begin{tabular}{cccccc}
\toprule
Observed $Y$ & $\hat{p}$(NoArrest$|x$) & $\hat{p}$(NonViolent$|x$) & $\hat{p}$(Violent$|x$) &Conformal Score \\ 
\hline \hline 
NoArrest & 0.58 & 0.25 & 0.17 & 1.0 \\
NonViolent & 0.49 & 0.32 & 0.18 & 0.50 \\
NoArrest & 0.47 & 0.35 & 0.18 &1.0 \\
Violent & 0.18 & 0.26 & 0.58 & 1.0 \\
NonViolent & 0.41 & 0.37 & 0.21 & 0.59 \\
NoArrest & 0.27 & 0.54 & 0.19 & 0.46 \\
\hline
\end{tabular}
\end{center}
\label{tab:conformalscore}
\end{table}

From Step 2 in the algorithm, Table~\ref{tab:conformalscore} shows, for a random subset of six $D_2$ cases, the construction of conformal scores as described in Equation~\eqref{eq:conformal-scores-ordering}. The far left column contains the observed outcome class: ``NoArrest'' for no arrest, ``NonViolent'' for an arrest for a nonviolent crime, and ``Violent'' for an arrest for a violent crime. The three columns to the right contain the estimated outcome class probabilities case by case. Conformal scores are provided on the far right. For example, the offender shown at the bottom of the table had no re-arrest while on probation, yet an arrest for a nonviolent crime was forecasted by the classifier because that outcome class had the largest estimated probability. Therefore, the conformal score becomes $0.27 + 0.19 = 0.46$; see Equation~\eqref{eq:example-conformal-score} and the following discussion. 
\begin{figure}[htbp]
% \vspace{-0.16in}
\begin{center}
\includegraphics[scale=.50]{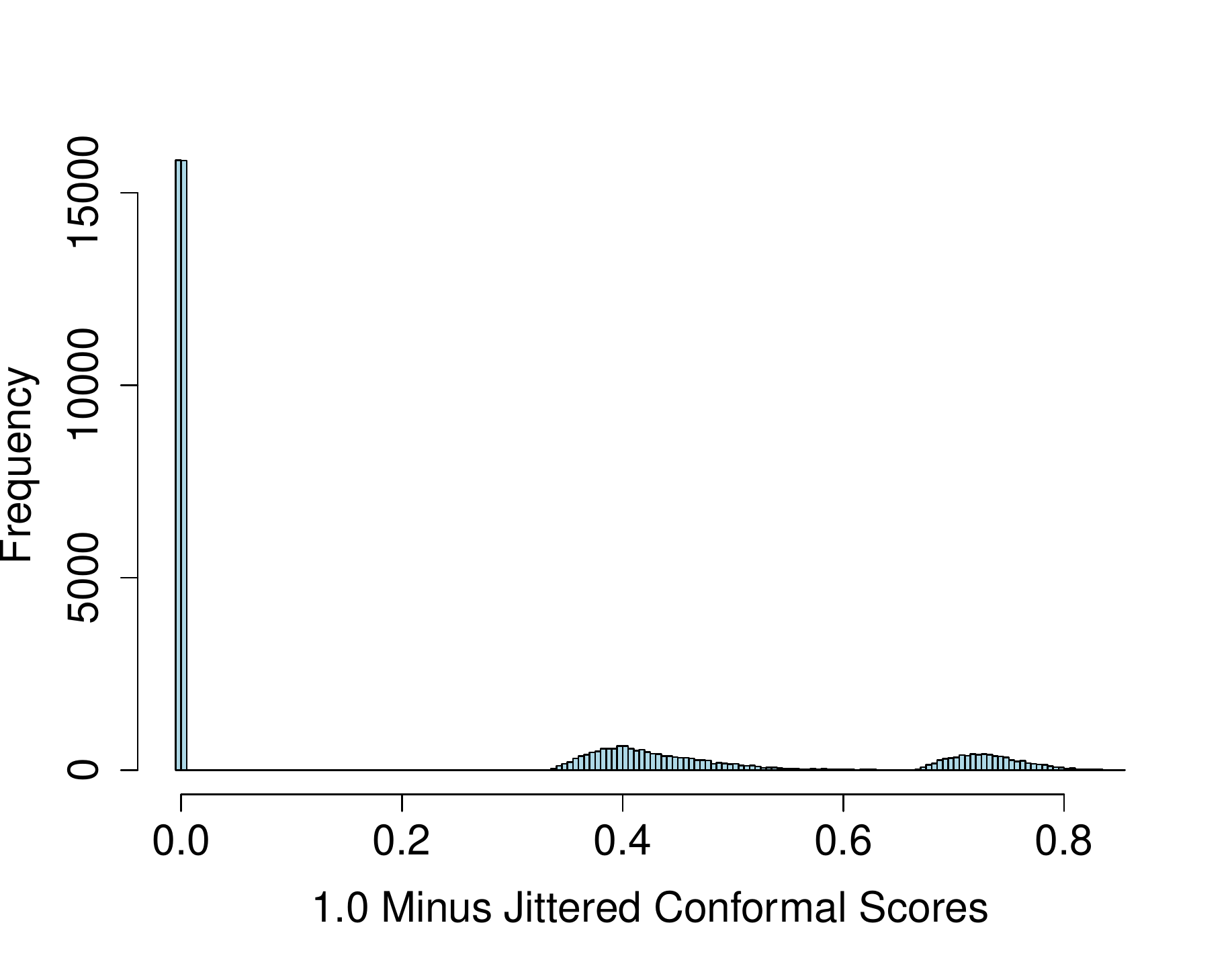}
\caption{Historgram of nested conformal scores from the $D_2$ (test) data.}
\label{fig:conhist}
\end{center}
\end{figure}
Figure~\ref{fig:conhist} is a histogram of the conformal scores $s(X_i, Y_i), i\in D_2$, transformed into non-conformal scores $1 - s(X_i, Y_i), i\in D_2$. There are three distinct of groups of scores. On the far left are the scores for cases in which the class with the highest probability corresponds to the observed outcome class. These are the cases that conform the most closely to the trained classifier and hence, have a non-conformity score of $0$. The scores in the middle are for cases in which the observed outcome class has the second largest fitted probability. The scores to the far right are for cases in which the observed outcome class had the smallest fitted probability. The degree of conformity declines from left to right. The large spike at zero is consistent with the earlier results from the confusion table from the stochastic gradient boosting classifier and are also expected because the classifier should do better than chance if there is a signal in the data.

\begin{table}[htb]
\centering
\caption{Nested conformal prediction sets for four examples when $\alpha = 0.30$ and $\alpha = 0.05$. We use the designation $0 = \mbox{No Arrest},$ $1 = \mbox{NonViolent Arrest}$, and $2 = \mbox{Violent Arrest}$. The first column is the probabilities from the classifier, the second column is the conformal scores computed based on~\eqref{eq:ordering-probabilities}. The third and fourth columns show the resulting prediction sets from Algorithm~\ref{alg:nested-conformal-prediction} when $\alpha = 0.30$ and when $\alpha = 0.05$. When $\alpha = 0.30$, the calibrated $\widehat{\gamma}(\alpha)$ in step 3 of Algorithm~\ref{alg:nested-conformal-prediction} is $0.6$ and when $\alpha = 0.05$, $\widehat{\gamma}(\alpha) = 0.26$.}
\label{tab:examples}
\resizebox{\textwidth}{!}{
\begin{tabular}{cccc}
\toprule
\begin{tabular}[c]{@{}c@{}}Probabilities\\ $\widehat{p}(0|x), \widehat{p}(1|x), \widehat{p}(2|x)$\end{tabular}      & \begin{tabular}[c]{@{}c@{}}Conformal Scores\\ $s(x,0), s(x, 1), s(x, 2)$\end{tabular}            & \begin{tabular}[c]{@{}c@{}}Prediction set\\$\alpha = 0.3$, $\widehat{\gamma}(\alpha) = 0.6$\end{tabular} & \begin{tabular}[c]{@{}c@{}}Prediction set\\$\alpha = 0.05$, $\widehat{\gamma}(\alpha) = 0.26$\end{tabular} \\ \hline\hline
(0.34,\, 0.27,\, 0.38) & (0.61,\, 0.27,\, 1.00) & $\{2, 0\}$                                                                                             & $\{2, 0, 1\}$                                                                                   \\% \hline
(0.19,\, 0.24,\, 0.56) & (0.19,\, 0.44,\, 1.00) & $\{2\}$                                                                                                        & $\{2, 1\}$                                                                                              \\% \hline
(0.60,\, 0.24,\, 0.15) & (1.00,\, 0.40,\, 0.15) & $\{0\}$                                                                                                      & $\{0, 1\}$                                                                                             \\ %\hline
(0.30,\, 0.35,\, 0.34) & (0.30,\, 1.00,\, 0.65) & $\{1, 2\}$                                                                                            & $\{1, 2, 0\}$                                                                                    \\ \hline
\end{tabular}}
\end{table}
% \begin{table}[htp]
% \tiny
% \caption{Four Examples of Prediction Set Construction: Estimated Class Probabilities, Class Conformal Scores, and Prediction Sets}
% \begin{center}
% \begin{tabular}{|l|c|c|c|c|c|}
% \hline
% Statistic & Case 1 & Case 2 & Case 3 & Case 4\\
% \hline \hline
% $\hat{p}$(No Arrest) &  0.34 & 0.19 & 0.60 & 0.30 \\
% $\hat{p}$(Nonviolent Arrest) & 0.27 & 0.24 & 0.24 & 0.35 \\
% $\hat{p}$(Violent Arrest) & 0.38 & 0.56 & 0.15 & 0.34 \\
% \hline
% Score No Arrest  & 0.61 & 0.19 & 1.0 & 0.30\\
% Score Nonviolent Arrest  & 0.27 & 0.44 & 0.40 & 1.0\\
% Score Violent Arrest & 1.0 & 1.0 & 0.15 & 0.65\\
% \hline
% Prediction Set for $\alpha=.30$, $\hat{\gamma}=0.60$  & Viol, None  & Viol & None & NonViol, Viol \\
% Prediction Set for $\alpha=.05$, $\hat{\gamma}=0.26$ & Viol, None, NonViol & Viol, NonViol & None, NonViol & NonViol, Viol, None\\
% \hline
% \end{tabular}
% \end{center}
% \label{tab:examples}
% \end{table}

Algorithm steps 3 and 4 are illustrated in Table~\ref{tab:examples} using a random set of four $D_2$ cases as if the true outcome were unknown. %The task is forecasting. Assume for now the forecasts are unconditional. 

With $\alpha=.05$, a high probability coverage of 0.95 is being sought. The $\hat{\gamma} (\alpha)$ of 0.26 is relatively small making it more likely that a larger number of outcome classes are included in a prediction set. The column on the far right shows two prediction sets with two elements and two prediction sets with three elements. For $\alpha =.30$ to the immediate left, a moderate coverage probability of $0.70$ is preferred. The $\hat{\gamma}(\alpha) = 0.60$, which is a more demanding criterion for precision. There are two prediction sets with two elements and two prediction sets with a single element. A larger value of $\alpha$ favors smaller prediction sets with smaller probability guarantees. One has greater precision but less certainty.

Because the value of $\alpha$ is chosen by the data analyst, it offers some control over the properties of the forecasts. Should prediction sets with fewer elements be preferred at the cost of smaller coverage probabilities? For criminal justice applications such as ours, the answer will come from stakeholders. For example, stakeholders might be especially concerned about offenders who are forecasted to commit violence crimes. Nonviolent crimes might be viewed primarily as a public nuisance, undesirable to be sure, but not worth increasing incarceration rates. A larger value of $\alpha$ might then follow. Yet, some stakeholders may be uncomfortable with greater uncertainty. For them, a smaller value of $\alpha$ might be preferred even if the prediction sets are larger. Such tradeoffs must be addressed and resolved before a risk algorithm is deployed.

These issues apply equally to forecasts with a conditional guarantee~\eqref{eq:conditional-guarantee}.
In practice, conditional forecasts can be desirable because probation decisions are made one individual at a time. For instance, if re-arrests differ on the average between male and female offenders, conditional forecasts offer a more responsive, and arguably more fair, approach. But the coverage guarantees now are only asymptotic. %For table~\ref{tab:examples}, the large number of D1 and $D_2$ cases in the probation data can justify conditional prediction sets.
\medskip

\noindent\textbf{Comparison of Conformal Prediction and Confusion table.} In previous sections, we have presented two different methods of understanding forecasts. With the confusion table (Table~\ref{tab:confusion-table}), we obtain estimated probabilities for when the true outcome matches the best forecast from the classifier (i.e., the outcome with the largest outcome probability). But the information we obtain from conformal inference conveys more complete information. 

First, our unconditional estimates are valid in finite samples. They also are agnostic to the choice of classifier. Such flexibility can broaden usefully the range of appropriate applications.

Second, our estimates can offer forecasts adapted to different subgroups defined by their predictor values. Stakeholders often note that there is considerable heterogeneity in the backgrounds of convicted offenders. One size should not be permitted fit all. 

% Third, we are able to provide a ``can't tell'' forecast when single outcome class is not determined by the classifier to be sufficiently more probable than the rest. ``Can't tell'' can impart very useful information that, for instance, can eliminate some improbable outcome classes or motivate decision makers to seek other sources of information. It also is arguably a more realistic and forthcoming way to forecast risk. 

Nevertheless, there can be arguments for relying solely on confusion tables constructed with test data. They provide probability estimates for subgroup of subjects determined by their forecasted class. Here, that would be ``NoArrest'', or ``NoViolent'', or ``Violent.'' For example, from Table~\ref{tab:confusion-table}, we know that if the forecast is ``NoArrest,'' the estimated the probability that ``NoArrest'' is the true outcome is $0.81$ (i.e., $1 - 0.19$). One might argue that such a prediction set with a single outcome is sound when $\alpha = 0.20$ and the forecast is ``NoArrest.''
Moreover, because the nested conformal prediction set in Algorithm~\eqref{alg:nested-conformal-prediction} does not condition on a single forecasted outcome class, useful information is lost, and conformal methods obtain the same threshold $\widehat{\gamma}(\alpha)$ for all subjects, regardless of their forecast. 

These distinguishing features of the confusion table and nested conformal prediction can sometimes lead to conflicting conclusions. For example, should stakeholders favor $\alpha = 0.2$, then for a subject with a forecast of ``NoArrest,'' a confusion table might convey that the truth is ``NoArrest'' at the required probability, but the nested conformal prediction set could report a prediction set with more than one outcome for the same required probability. 

A way to avoid such discrepancies is to estimate nested conformal prediction sets separately on groups, each defined by their forecast. We turn, therefore, to \emph{local} nested conformal prediction regions in the following section.

\section{Localized Conformal Prediction}\label{sec:local-conformal-prediction}
Should there be a desire to match better forecasting conventions used by classifiers, the test data can be separated into subsets, each determined by the outcome class with the largest probability. Such subsets are the same as those in each column of a conventional confusion table, but the nested conformal procedure is applied to produce prediction sets. For the probation data, one would apply the nested conformal procedure separately three times: on the subset of test data determined by the classifier's ``NoArrest'' forecast, on the subset of test data determined the classifier's ``NonViolent'' forecast, and on the subset of test determined by the classifier's ``Violent'' forecast.

We call this method ``localized conformal prediction,'' summarized in Figure~\ref{fig:localized_conformal_prediction} for the probation data.
% \begin{figure}[!h]
%     \centering
%     \includegraphics[width = \textwidth]{Localized.pdf}
%     \caption{Illustration of Localized Conformal Prediction With Each Test Data Subset Determined by the Classifier's Forecast}
%     \label{fig:localized_conformal_prediction}
% \end{figure}
\begin{figure}[!h]
    \centering
    \includegraphics[width=\textwidth]{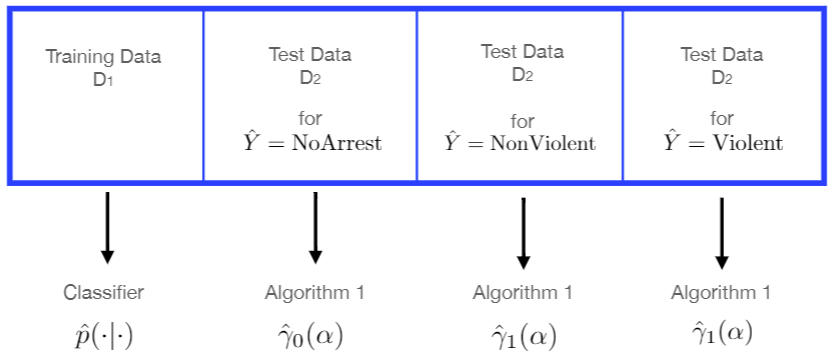}
    \caption{Illustration of Localized Conformal Prediction With Each Test Data Subset Determined by the Classifier's Forecast}
    \label{fig:localized_conformal_prediction}
\end{figure}
% \begin{figure}[!h]
%     \centering
%     \includegraphics[width = 5in \textwidth]{Localized.cropped.pdf}
%     \caption{Illustration of Localized Conformal Prediction With Each Test Data Subset Determined by the Classifier's Forecast}
%     \label{fig:localized_conformal_prediction_arun}
% \end{figure}
Algorithm~\ref{alg:nested-conformal-prediction} is applied three times, once to each test data subset, to produce three thresholds $\widehat{\gamma}_0(\alpha), \widehat{\gamma}_1(\alpha),$ and $\widehat{\gamma}_2(\alpha).$ The three potential prediction sets denoted by $j$ are,
\[
\widehat{C}_{j,\alpha}(x) = \{y\in\{0,1,2\}:\, s(x, y) \ge \widehat{\gamma}_j(\alpha)\},\quad\mbox{for}\quad j \in \{0, 1, 2\}.
\]

For any future subject with the covariate vector $x$, we first use the classifier from the training data to obtain the forecast $\widehat{Y}$. If $\widehat{Y} = 0$, one should report the prediction set $\widehat{C}_{0,\alpha}(x)$.\footnote{Here again we use the designation: ``0 = NoArrest'', ``1 = NonViolent'', and ``2 = Violent''.} If $\widehat{Y} = 1$, one should report the prediction set $\widehat{C}_{1,\alpha}(x)$. If $\widehat{Y} = 2$, one should report the prediction set $\widehat{C}_{2,\alpha}(x)$. That is, the appropriate prediction set is $\widehat{C}_{\widehat{Y},\alpha}(x)$.

The prediction set constructed in this manner is more selective than the output of Algorithm~\ref{alg:nested-conformal-prediction} applied on the complete test data $D_2$. Nevertheless, the prediction set $\widehat{C}_{\widehat{Y},\alpha}(x)$ retains finite sample unconditional guarantee~\eqref{eq:unconditional-guarantee}, irrespective of the accuracy of the classifier $\widehat{p}(\cdot|\cdot)$. 
\begin{thm}
If the data $(X_i, Y_i), i\in D_2$ are exchangeable, then the prediction set $\widehat{C}_{\widehat{Y},\alpha}(\cdot)$ obtained consistent with Figure~\ref{fig:localized_conformal_prediction} satisfies
\[
\mathbb{P}\left(Y_{\mathrm{f}} \in \widehat{C}_{\widehat{Y}_{\mathrm{f}},\alpha}(X_{\mathrm{f}})\right) \ge 1 - \alpha.
\]
\end{thm}
This result follows from the arguments of Lemma 1 of~\cite{romano2019malice}; also, see Section 4 of~\cite{vovk2012conditional}. The localized conformal prediction algorithm is one of the many variants of Algorithm~\ref{alg:nested-conformal-prediction} leading to finite sample prediction sets that can be more responsive to particular applications. Such variants can be found in~\cite{hechtlinger2018cautious},~\cite{sadinle2019least}, and~\cite{guan2019prediction}. We introduce the localized formulation because the localized data structure corresponds well to the usual inferential framework for a confusion table. 

\medskip

\noindent\textbf{Application of Localized Conformal Method to Probation Data.} To illustrate localized conformal method, we consider its performance compared to the conventional confusion table approach. The probation data serve as as a testbed, and there are implications for future practice.

Working with a usual confusion table, suppose a data analyst constructs prediction sets containing solely the high probability forecast. This follows from the test data subsetting described above.  Forecasting errors~\eqref{eq:forecasting-error} provide the miscoverage probabilities conditional on the forecast. For instance, from the last row of Table~\ref{tab:confusion-table}, when the forecast is no arrest, the true outcome is no arrest with an estimated probability of $1 - 0.19 = 0.81$. Consequently, reporting a single element prediction set would lead to a $81\%$ coverage when the true forecast is no arrest. But for the cases with a forecast of non-violent arrest, the single element prediction set has a coverage of $1-0.47=0.53$. Single element prediction set can have different coverage across different cases. In summary, confusion tables fix the number of elements in the prediction set to one and lead to different coverage levels. 

Conformal prediction methods, in contrast, fix coverage at a desired level of $1 - \alpha$ and find a set of one or more outcome classes such that coverage is at least $1 - \alpha$. The nested conformal prediction method (described in Algorithm~\ref{alg:nested-conformal-prediction}) does not distinguish between cases depending on their forecast. The localized conformal prediction method (illustrated in Figure~\ref{fig:localized_conformal_prediction}) distinguishes between cases depending on their forecast, similar to the confusion table. This is the key distinction between these two conformal methods. 
\newcolumntype{g}{>{\columncolor{Gray}}c}
\begin{table}[htb]
\centering
\caption{Comparison of nested conformal and localized conformal prediction sets for the four examples in Table~\ref{tab:examples} when $\alpha = 0.30$ and $\alpha = 0.05$. We use the designation $0 = \mbox{No Arrest},$ $1 = \mbox{NonViolent Arrest}$, and $2 = \mbox{Violent Arrest}$. The first and third columns show the nested conformal prediction sets from Algorithm~\ref{alg:nested-conformal-prediction} when $\alpha = 0.30$ and when $\alpha = 0.05$; these are repeated from Table~\ref{tab:examples}. The second and fourth columns show the localized conformal prediction sets from Figure~\ref{fig:localized_conformal_prediction} when $\alpha = 0.30$ and $\alpha = 0.05$, respectively. The threshold $\widehat{\gamma}(\alpha)$ for nested conformal are mentioned in Table~\ref{tab:examples}. There are three thresholds for each $\alpha$ for localized conformal. For $\alpha = 0.3$, $\widehat{\gamma}_0(\alpha)=0.99, \widehat{\gamma}_1(\alpha)=0.58, \widehat{\gamma}_2(\alpha)=0.32$. For $\alpha=0.05$, $\widehat{\gamma}_0(\alpha)=0.45, \widehat{\gamma}_1(\alpha)=0.26, \widehat{\gamma}_2(\alpha)=0.24.$}
\label{tab:localized-conformal-examples}
\resizebox{\textwidth}{!}{
\begin{tabular}{cgcg}
\hline
\begin{tabular}[c]{@{}c@{}}Prediction set\\ when $\alpha = 0.3$\end{tabular} & \begin{tabular}[c]{@{}c@{}}Localized Prediction\\set when $\alpha = 0.3$\end{tabular} & \begin{tabular}[c]{@{}c@{}}Prediction set\\ when $\alpha = 0.05$\end{tabular} & \begin{tabular}[c]{@{}c@{}}Localized Prediction\\set when $\alpha = 0.05$\end{tabular} \\ \hline\hline
\{2,0\}                                                                                                        & \{2,0\}                                                                                & \{2, 0, 1\}                                                                                                      & \{2, 0, 1\}                                                                             \\ \hline
\{2\}                                                                                                          & \{2,1\}                                                                                & \{2, 1\}                                                                                                         & \{2, 1\}                                                                                \\ \hline
\{0\}                                                                                                          & \{0\}                                                                                  & \{0, 1\}                                                                                                         & \{0\}                                                                                   \\ \hline
\{1, 2\}                                                                                                       & \{1, 2\}                                                                               & \{1, 2, 0\}                                                                                                      & \{1, 2, 0\}                                                                             \\ \hline
\end{tabular}}
\end{table}
% For our data, there is substantial variation in the outcome classes; a large proportion (60\%) of the true outcomes are ``No Arrest''. {\color{red}{I don't see how this sentence applies to the rest of the paragraph.}} 
In Table~\ref{tab:localized-conformal-examples}, we compare the nested conformal and localized conformal prediction sets for the four example cases shown in Table~\ref{tab:examples}. At a lower level of confidence ($0.7 = 1 - 0.3$), conformal and localized prediction sets are no different, for the four examples presented. At a higher level of confidence ($0.95 = 1 - 0.05$), localized prediction sets are smaller than conformal prediction sets, when the forecast is ``No Arrest'' (or $0$). Precision differs in this instance, and in the probation policy setting, a prediction set with a single element can substantially simplify subsequent decisions. 
\begin{table}[!h]
\centering
\caption{Proportions of obtained prediction sets of particular sizes for nested and localized conformal method when $\alpha = 0.3$ and $\alpha = 0.05$. For both nested and localized conformal methods, we computed the prediction sets for 51,277 test data points when $\alpha = 0.3$ and $\alpha = 0.05$. When $\widehat{Y} = 0$, both methods can only report prediction sets $\{0\}, \{0, 1\}, \{0, 2\}$ and $\{0, 1, 2\}$. Similarly, for $\widehat{Y} = 1$, the possible prediction sets are $\{1\}, \{1, 0\}, \{1, 2\}$ and $\{1, 0, 2\}$. For each forecast and each method, we report the proportion of cases for which the obtained prediction set has one element (denoted by ``One''), two elements (denoted by ``Two''), and three elements (denoted by ``Three''). Singletons (sets with one outcome) are the most informative, two element set are the next most informative, and finally, three element sets are the least informative. The proportions are rounded-off to three digits for clarity.}
\label{tab:proportion-of-sets-2}
\begin{tabular}{ccccggg}
\hline
                  & \multicolumn{3}{c}{\begin{tabular}[c]{@{}c@{}}Nested Conformal\\ ($\alpha = 0.3$)\end{tabular}}  & \multicolumn{3}{c}{\begin{tabular}[c]{@{}c@{}}Localized Conformal\\ ($\alpha = 0.3$)\end{tabular}}  \\ \hline
                  & One                             & Two                            & Three                         & One                              & Two                             & Three                          \\ \toprule
$\widehat{Y} = 0$ & $0.835$                         & $0.165$                        & $0.000$                       & $1.000$                          & $0.000$                         & $0.000$                        \\
$\widehat{Y} = 1$ & $0.720$                         & $0.280$                        & $0.000$                       & $0.508$                          & $0.492$                         & $0.000$                        \\
$\widehat{Y} = 2$ & $0.624$                         & $0.376$                        & $0.000$                       & $0.000$                          & $0.982$                         & $0.018$                        \\ \hline
                  & \multicolumn{3}{c}{\begin{tabular}[c]{@{}c@{}}Nested Conformal\\ ($\alpha = 0.05$)\end{tabular}} & \multicolumn{3}{c}{\begin{tabular}[c]{@{}c@{}}Localized Conformal\\ ($\alpha = 0.05$)\end{tabular}} \\ \hline
                  & One                             & Two                            & Three                         & One                              & Two                             & Three                          \\ \toprule
$\widehat{Y} = 0$ & $0.000$                         & $0.798$                        & $0.202$                       & $0.200$                          & $0.800$                         & $0.000$                        \\
$\widehat{Y} = 1$ & $0.000$                         & $0.627$                        & $0.373$                       & $0.000$                          & $0.602$                         & $0.398$                        \\
$\widehat{Y} = 2$ & $0.000$                         & $0.510$                        & $0.490$                       & $0.0$                            & $0.274$                         & $0.726$                         \\ \hline               
\end{tabular}
\end{table}

Informally, the localized conformal method accounts for the heterogeneity of outcomes in the data and can yield greater precision than the nested conformal method for some outcome classes. 

A more extensive comparison is provided in Table~\ref{tab:proportion-of-sets-2}. In this table, we report the proportion of cases (out of each forecasted outcome) for which the conformal and localized methods lead to a prediction set with one element, two elements, or three elements. For instance, the first row of Table~\ref{tab:proportion-of-sets-2} shows the proportion of test subjects for whom the highest probability forecast is $0$ (= ``No Arrest''). So, the first number $0.835$ is the proportion of subjects for whom the nested conformal method returned a single element prediction set among the subjects for whom the highest probability forecast is $0$ (= ``No Arrest''). Prediction sets with a single member are desirable, especially in policy settings when decisions need to be made. 

There is a lot to consider in the table. First, there is the earlier tradeoff between $1-\alpha$ and precision. Comparing the top panel to the bottom panel, prediction sets with a single element are far more common in the top panel for which $1-\alpha = 0.70$. Prediction sets with three elements are far more common in the bottom panel for which $1-\alpha = 0.95$. These patterns, which dominate Table~\ref{tab:proportion-of-sets-2}, make it difficult to isolate the impact of the conformal method (i.e., the columns).

Within either the top or bottom panel, there are not for this analysis consistent differences in precision between the two conformal methods. Sometimes the nested conformal approach is more precise, and sometimes the localized conformal approach is more precise. For forecasts of no re-arrest, localized conformal inference performs better. For example, in the top panel, the localized method always arrives at prediction sets with a single element. But for the two less common outcomes, the nested conformal approach arguably is superior, particularly for a violent crime arrest, the least most common outcome classes.

% These differences in precision can be linked to the cost ratios used in this analysis. 
These differences in precision can be linked to the difference in the guarantees provided by the nested and localized conformal methods. Recall from Figure~\ref{fig:localized_conformal_prediction} that the localized conformal approach applies the nested conformal algorithm separately at the same  $1-\alpha$ level to three subsets of test data, each corresponding to the outcome forecasted. Concomitantly, the localized conformal approach provides a coverage guarantee of at least $1-\alpha$ for each of these subsets, despite far smaller numbers of observations for the two arrest outcome classes from which prediction sets are derived. Furthermore, the stakeholder cost ratios in this analysis introduce more error into forecasts for both of the two arrest outcomes, especially for violent crime arrests. The cost ratios incentivize the risk algorithm to especially avoid misclassifying offenders who are at high risk for re-arrest. One result is that for these outcome classes, there are many more false positives, which introduce greater forecasting error. The localized conformal approach is undermined by these false positives when working with one subset of the data at a time. The nested approach, which uses all the data at once, is not disadvantaged in this manner. The difference of guaranteeing coverage on each subsets versus on all of test data makes it harder to directly compare localized and nested conformal methods; with a stronger guarantee of localized conformal method comes a larger prediction set.

\section{Conclusions}\label{sec:conclusions}

Aggregate performance evaluations made from the earlier confusion table constructed from test data, made a case that the algorithmic risk assessments had better accuracy than the means by which probationer risk levels had previously been determined. If supervision intensity were substantially informed by the algorithmic forecasts instead of current practice, many forecasting errors could be eliminated. Those errors have consequences for supervising offenders with the least restrictive means consistent with public safety and crime control. But this is old news addressed in \cite{barnes2010low, berk2010second}. 

The new news is that conformal prediction sets have the capacity to improve the correspondence between risk forecasts and needs of stakeholders.  Confusion tables constructed from test data and their summary statistics can play an essential role in evaluating the overall performance of algorithmic risk procedures. Classification error rates, for example, are one important metric. But forecasts for individuals that inform particular decisions are better represented as prediction sets than the single most likely outcome. There will often be times when two or more outcome classes have nearly the same outcome probabilities. Forcing a single choice, even if desirable from a policy perspective, can misrepresent the information provided. In contrast, when a prediction set includes a single element, the classification algorithm has been able to make a relatively definitive and defensible assessment of risk.

Conformal prediction sets come with additional desirable features. There is a forecasted prediction set and the probability that the true outcome class is an element in that prediction set is controlled as the user specifies. One has information on both precision and uncertainty. Also, data analysts can make tradeoffs between precision and uncertainty allowing algorithmic risk assessments to be more responsive to the policy setting. 

Nested conformal prediction sets address precision and uncertainty with finite sample guarantees, but conditional forecasts also are available with asymptotic guarantees. Conditional forecasts can be especially useful is criminal justice risk assessments because the forecasts can apply to similarly situated offenders. Finally, if as a policy matter a refined forecast is needed, localized conformal prediction sets can be constructed that condition on the forecasted outcome class. Localization requires that the confusion tables and conformal prediction sets work with the same subsets of cases. 

In practice, all algorithmic risk results are only as good as the exchangeable data available and the forecasting skill of the algorithm. Criminal justice administrative data that are routinely accessible will often suffice, although their quality and depth can vary by jurisdiction and the criminal justice decision to be made. For example, if an offender has been incarcerated in a state prison, the available data may be especially rich. The same generally does not apply to incarceration in a county jail. But whatever the data may be, credible arguments must be made to justify exchangeability. 

Risk algorithms vary widely in feasibility, accuracy, speed, and ease of use. A discussion is beyond the scope of this paper. Suffice it to say that there can be challenging tradeoffs and that proper tuning is essential. \cite{hastie2009elements} provide an instructive discussion. 
% {\color{red}(Trevor Hastie, Robert Tibshiranio, and Jerome Friedman, 2009, The Elements of Statistical Learning, second edition, Springer)}

It remains to be seen whether conformal prediction methods can be accepted by criminal justice decision makers and stakeholders, should there be an opportunity to improve on current parole risk procedures. Refitting the classifier with current data would surely be a major improvement. More emphasis on capturing and accurately communicating uncertainty also would be very useful. However, there is considerable resistance in general to algorithmic risk forecasting as currently practiced or proposed, fueled significantly by misunderstandings about the procedures. Fairness and transparency matter too. These complications are locally salient and will shape prospects for any revisions of the parole risk methods. At the very least, there is important educational work ahead.  

\bibliographystyle{apalike}
\bibliography{conformal}
\end{document}